\documentclass{dalthesis}
\usepackage{amsmath,amsfonts,amsthm,amssymb}
\usepackage{setspace}
\usepackage{enumerate}
\usepackage{fancyhdr}
\usepackage{lastpage}
\usepackage{extramarks}
\usepackage{chngpage}
\usepackage{soul}
\usepackage[usenames,dvipsnames]{color}
\usepackage{graphicx,float,wrapfig}
\usepackage{ifthen}
\usepackage{listings}
\usepackage{courier}
\usepackage{algorithmic}
\usepackage{verbatim}
\usepackage{hyperref}
\hypersetup{pdfencoding=unicode}
\usepackage{inputenc}
\usepackage[all]{xy}
\usepackage{blkarray}
\usepackage{tikz}
\usetikzlibrary{arrows}

\theoremstyle{theorem}
\newtheorem{theorem}{Theorem}

\newtheorem{lemma}[theorem]{Lemma}

\theoremstyle{definition}
\newtheorem{algorithm}[theorem]{Algorithm}

\newtheorem{definition}[theorem]{Definition}

\theoremstyle{definition}  

\newtheorem{case}[theorem]{Case}

\newtheorem{example}[theorem]{Example}

\newtheorem{remark}[theorem]{Remark}

\theoremstyle{definition}
\newtheorem{convention}{Convention}

\newcommand{\Z}{\mathbb{Z}}
\newcommand{\N}{\mathbb{N}}
\newcommand{\D}{\mathbb{D}}

\newcommand{\sem}[1]{[\![#1]\!]}

\newcommand{\CNOT}{\left[\begin{array}{cccc}
1&0&0&0\\
0&1&0&0\\
0&0&0&1\\
0&0&1&0\\
\end{array}\right]}
\newcommand{\Hadamard}{\frac{1}{\sqrt{2}}\left[\begin{array}{rr}
1&1\\
1&-1\\
\end{array}\right]}
\newcommand{\Phase}{\left[\begin{array}{cc}
1&0\\
0&i\\
\end{array}\right]}

\newcommand{\jay}{j}
\newcommand{\kay}{k}
\newcommand{\tee}{t}
\newcommand{\ay}{a}
\newcommand{\bee}{b}

\renewcommand{\:}{\mathbin{:}}

\makeatletter
\newcommand{\xRightarrow}[2][]{\ext@arrow
  0359\Rightarrowfill@{#1}{#2}}
\makeatother
\newcommand{\level}{\mathop{\rm level}\nolimits}
\newcommand{\extent}{\mathop{\rm extent}\nolimits}
\newcommand{\lde}{\mathop{\rm lde}\nolimits}
\newcommand{\edge}{\xrightarrow}
\newcommand{\nedge}{\xRightarrow}
\renewcommand{\vec}{\overrightarrow}
\newcommand{\s}[1]{\{#1\}}
\newcommand{\mymod}{\mathop{{\rm mod}}}
\newcommand{\mmod}[1]{\,(\mymod#1)}
\newcommand{\divides}{\mid}

\newcounter{mycounter}
\def\myitem{\refstepcounter{mycounter}(\arabic{mycounter})}

\newcounter{casecounter}
\newcounter{subcasecounter}[casecounter]
\newcounter{subsubcasecounter}[subcasecounter]
\renewcommand{\thecasecounter}{\arabic{casecounter}}
\renewcommand{\thesubcasecounter}{\arabic{casecounter}.\arabic{subcasecounter}}
\renewcommand{\thesubsubcasecounter}{\arabic{casecounter}.\arabic{subcasecounter}.\arabic{subsubcasecounter}}
\def\case{\refstepcounter{casecounter}\item[{\bf Case \thecasecounter.}]}
\def\subcase{\refstepcounter{subcasecounter}\item[{\bf Subcase \thesubcasecounter.}]}
\def\subsubcase{\refstepcounter{subsubcasecounter}\item[{\bf Subcase \thesubsubcasecounter.}]}

\newcommand{\matindex}[1]{{\scriptsize#1}}

\newcommand{\urlalt}[2]{\href{#2}{\nolinkurl{#1}}}

\newcommand{\doi}[1]{\urlalt{doi:#1}{http://dx.doi.org/#1}}

\nosignaturepage
\nopermissionpage

\inputencoding{utf8}

\begin{document}
\title{Generators and Relations for the Group $U_4(\Z[\frac{1}{\sqrt2},\lowercase{i}])$}
\author{Seth Eveson Murray Greylyn}
\defencemonth{August}\defenceyear{2014}\copyrightyear{2014}
\convocation{October}{2014}
\supervisor{Peter Selinger}
\reader{Jason Brown}
\reader{Karl Dilcher}
\dedicate{For Jennifer.}
\frontmatter

\begin{abstract}
We give a presentation by generators and relations of the group $U_4(\Z[\frac{1}{\sqrt2},i])$ of unitary $4\times 4$ matrices with entries in the ring $\Z[\frac{1}{\sqrt2},i]$. This is motivated by the problem of exact synthesis for the Clifford+$T$ gate set in quantum computation.
\end{abstract}

\begin{acknowledgements}
I would like to acknowledge the guidance of my thesis supervisor, Dr. Peter Selinger, without whose skill, dedication, and patient advice this thesis would not have been possible. Special thanks should also be paid to Dr. Jason I. Brown and Dr. Karl Dilcher for graciously agreeing to act as readers for this work. The faculty and staff of Dalhousie University have provided me with too many mentors, confidantes, and friends to list exhaustively, but I should mention Dr. Sara Faridi and Paula Flemming, without whose counsel and assistance I could not have fulfilled the requirements of the degree for which this work was completed. Perhaps most importantly, I am indebted to Dr. Suzanne Seager for instilling in me the skill and confidence to recognise just how beautiful mathematics can be.
\end{acknowledgements}
\mainmatter
\chapter{Introduction}

In this work, we give a presentation of the group $U_4(\Z[\frac{1}{\sqrt{2}},i])$ in terms of generators and relations. This is motivated by the problem of {\em exact synthesis} in quantum computation. We shall use the rest of this introduction to very briefly sketch the main ideas behind the exact synthesis problem. Apart from the introduction, the remainder of this thesis is formulated in completely algebraic terms, and can be read without any background in quantum computation.

The problem in quantum information theory which led to the current work is the decomposition of unitary operators, via the operations of matrix multiplication and tensor product, into operators from a given finite set, called \textit{gates}. The decomposition can either be done exactly (`exact synthesis') or up to $\epsilon$ (`approximate synthesis'). A common gate set is the so-called Clifford+$T$ set, consisting of operators $\omega, H, S, T, CNOT:$\[
\omega = e^{i\pi/4},\quad H = \Hadamard,\quad S = \Phase,\quad T =\left[\begin{array}{cc}
1&0\\
0&\omega
\end{array}
\right],\]\[ CNOT = \CNOT.
\]\label{page-intro}
Note: since $S = T^2,$ it is not strictly necessary to include $S$ among the finite gate set; however, the $T$ gate is often considered more computationally `expensive' to perform than the $S$ gate, which explains why $S$ is commonly included in the gate set.

Also note that the gate set contains matrices of different dimensions. Since we have closed the set of operators not only under matrix multiplication, but also under tensor product, this means that operators such as 
\[ I\otimes H = \left[\begin{array}{c|c}H&0\\\hline0&H\end{array}\right]
\]
are also included. Using terminology from quantum computing, a $2\times2$ matrix is called a \textit{single qubit} operator, a $4\times4$ matrix is called a \textit{two qubit} operator, and, more generally, a $2^n\times2^n$ matrix is called an \textit{$n$-qubit} operator. An operator that is explicitly built up from gates via tensor product and matrix multiplication is also called a {\em circuit}.

To be suitable for quantum computation, a gate set must be \textit{dense}, in the sense that every unitary $2^n\times 2^n$ matrix can be approximated in the operator norm, up to a unit scalar and arbitrary $\epsilon>0,$ by an operator generated from the gate set. The Clifford+$T$ gate set is dense. Moreover, it is well-adapted to quantum error correction, making it an attractive gate set to use {\cite{Nielsen-Chuang}}.

In the last few years, starting with the work of Matsumoto and Amano {\cite{MATS-AMAN}}, tools from algebra and number theory have been applied to the exact synthesis problem. The connection arises from the fact that groups of quantum circuits can often be characterised in purely algebraic terms; e.g., as groups of unitary matrices over a given ring of algebraic numbers. 

The first quantum-algebraic result of interest to us is by Kliuchnikov et al. {\cite{Kliuchnikov-etal}}, who showed that the single qubit Clifford+$T$ group is precisely the group $U_2(\Z[\frac{1}{\sqrt{2}},i])$ of $2\times 2$ unitary matrices with entries in the ring $\Z[\frac{1}{\sqrt{2}},i]$. This result was generalised by Giles and Selinger {\cite{GILES-SEL}}, who showed that the group of $n$-qubit Clifford+$T$ circuits (using at most one \textit{ancilla} or \textit{scratch space qubit}) is precisely the group of $2^n\times 2^n$ unitary matrices with entries from that ring.

The decomposition of a given operator into gates from a finite gate set is not typically unique. In general, it would be desirable to find circuits that are as `small' or `short' as possible, because quantum circuits could then be implemented with fewer resources. This problem is, in general, hard; finding the shortest possible Clifford+$T$ circuit implementing a given operator is likely to be an NP-complete problem. However, we are naturally led to the problem of \textit{circuit simplification}; i.e., given a particular circuit, to find an equivalent shorter one. To that end, it would be useful to have an \textit{equational} characterisation of circuit equivalence, or in other words, to find a set of relations that are sufficient to transform any circuit into any other equivalent circuit.

The single-qubit Clifford+$T$ group already has a known equational representation, which follows from a result of Matsumoto and Amano {\cite{MATS-AMAN}}. The presentation is by the generators $\omega,H,S,T$ and the following equations: \[\begin{array}{ccc}
\omega^8 &=& 1\\
 H^2 &= &1\\
  S^4 &=& 1\\
  SHSHSH&=&\omega\\
  \omega H&=&H\omega\\
  \omega T&=&T\omega\\
  \omega S&=&S\omega\\
  T^2 &=& S\\
HSSHT\omega &=& THSSHS
\end{array}
\]
Here, the symbols $H,S,T$ are interpreted by the corresponding $2\times 2$ matrices shown on p.~\pageref{page-intro}, and the symbol $\omega$ is interpreted as the $2\times 2$ matrix $\omega I$.

In this thesis, we consider the next case (i.e., the case of two qubits), consisting of the group $U_4(\Z[\frac{1}{\sqrt{2}},i]).$ Our main result is a presentation of this group in terms of generators and relations. We aim to make the work as self-contained as possible. Where results from the literature are cited without proof, they are usually only provided for motivation and context. All essential tools are developed \textit{in situ}.

\chapter{Generators for
  \texorpdfstring{$\boldsymbol{U_n(\D[\omega])}$}{U^^e2^^82^^99(^^e2^^85^^85[^^cf^^89])}}

\section{Some Rings}

Throughout, we use the physics notations $x^\dagger$ to denote the
complex conjugate of complex number $x$, and $\boldsymbol{v}^\dagger$
to denote the adjoint of a vector.

\begin{definition}
\label{RingsDef}
Let $\omega = \frac{1+i}{\sqrt{2}} = e^{i\pi/4}.$ Note that $\omega$ is an eighth root of unity, and so in particular, the relation $\omega^8 = 1$ listed in the introduction is evident.

Let $\D$ be the ring of \textbf{dyadic fractions}; i.e., $\D = \Z[\frac{1}{2}] = \{\frac{a}{2^n}\:a\in\Z,n\in\N_0\}$. From these we construct two essential rings:
\begin{enumerate}
\item$\D[\omega] = \{a\omega^3+b\omega^2+c\omega+d\:a,b,c,d\in\D\}.$
\item$\Z[\omega] = \{a\omega^3+b\omega^2+c\omega+d\:a,b,c,d\in\Z\}.$
\end{enumerate}
\end{definition}

Note that, since $\frac12=(\frac1{\sqrt2})^2$ and $\omega=\frac{1+i}{\sqrt{2}}$ are elements of $\Z[\frac{1}{\sqrt{2}},i]$, and conversely, $\frac{1}{\sqrt{2}}=\frac{\omega-\omega^3}{2}$ and $i=\omega^2$ are elements of $\D[\omega]$, it follows that $\D[\omega] = \Z[\frac{1}{\sqrt{2}},i].$ Henceforth we shall use the less-cumbersome $\D[\omega]$ notation to refer to this ring. 

$\Z[\omega]$ is called the ring of \textit{cyclotomic integers of degree 8}. It is well-known that $\Z[\omega]$ is a Unique Factorization Domain {\cite{Washington}}. In particular, the notion of {\em divisibility} makes sense in this ring; we write $a\divides b$ if there exists some $c\in\Z[\omega]$ with $ac=b$, and $b\equiv c\pmod a$ of $a\divides(c-b)$.

\begin{definition}
\label{U2nDef}
Let $U_{n}(\D[\omega])$ be the group of all $n\times n$ unitary matrices with elements from $\D[\omega].$
\end{definition}

\begin{example}
\label{IntroExample}
The operators listed on p.~\pageref{page-intro}
are elements of $U_1(\D[\omega])$, $U_2(\D[\omega])$, and
$U_4(\D[\omega])$, according to their respective dimensions.
\end{example}

\section{One- and Two-Level Operators}

To find suitable generators for the group $U_n(\D[\omega]),$ we find it convenient to make use of the concepts of one- and two-level operators. Thus we make

\begin{definition}~\label{2LvlDef}
Let $U = \begin{bmatrix}
x_{1,1} \, x_{1,2} \\
x_{2,1} \, x_{2,2}
\end{bmatrix}$ and $1\leq\alpha<\beta\leq n$. Then \[
U_{[\alpha,\beta]} = \raisebox{1ex}{$\begin{blockarray}{cccccccc}
  &&\matindex{\cdots}&\matindex{\alpha}&\matindex{\cdots}&\matindex{\beta}&\matindex{\cdots}& \\
  \begin{block}{c[c@{}c|c|c|c|c@{}c]}
    \matindex{\vdots}&&I&0&0&0&0 \\ \cline{3-7}
    \matindex{\alpha}&&0&x_{1,1}&0&x_{1,2}&0\\ \cline{3-7}
    \matindex{\vdots}&&0&0&I&0&0\\ \cline{3-7}
    \matindex{\beta}&&0&x_{2,1}&0&x_{2,2}&0 \\ \cline{3-7}
    \matindex{\vdots}&&0&0&0&0&I \\
  \end{block}
\end{blockarray}$}.
\]
We call $U_{[\alpha,\beta]}$ a \textbf{two-level operator} which applies the operator $U$ to rows (or columns) $\alpha$ and $\beta$ of an $n\times n$ matrix.
\end{definition}

\begin{definition}
\label{1LvlDef}
A \textbf{one-level operator} is defined analogously to the above. Let $U = \left[x_1\right]$ be the $1\times 1$ matrix whose entry is $x_1,$ and let $1\leq\alpha\leq n.$ Then \[
U_{[\alpha]} = \raisebox{1ex}{$\begin{blockarray}{cccccc}
  &&\matindex{\cdots}&\matindex{\alpha}&\matindex{\cdots}& \\
  \begin{block}{c[c@{}c|c|c@{}c]}
    \matindex{\vdots}&&I&0&0 \\ \cline{3-5}
    \matindex{\alpha}&&0&x_{1}&0\\ \cline{3-5}
    \matindex{\vdots}&&0&0&I \\
  \end{block}
\end{blockarray}$}.
\]
\end{definition}

\section{Generators for \texorpdfstring{$U_n(\D[\omega])$}{U^^e2^^82^^99(^^e2^^85^^85[^^cf^^89])}}

\begin{definition}
\label{StarDef}
Let \[
X = \begin{bmatrix}
0&1\\
1&0
\end{bmatrix},\quad H = \Hadamard,\quad\mbox{ and }\omega = \left[e^{i\pi/4}\right].
\]
Consider the set of operators
\begin{equation}
 \{X_{[\alpha,\beta]},H_{[\alpha,\beta]},\omega_{[\alpha]}\mid 1\leq\alpha<\beta\leq n\}
\tag{$*$}
\end{equation}
Note that we have slightly abused the notation here by using the symbol $\omega$ to denote both a scalar and the corresponding $1\times 1$ matrix. The meaning will always be clear from the context.
\end{definition}

In fact, the above set of operators generates the group under consideration. This is established by the following theorem of Giles and Selinger, from \cite{GILES-SEL}.

\begin{theorem}
\label{GSDecompThm}
Let $M$ be an $n\times n$ matrix. Then $M\in U_{n}(\D[\omega])$ if, and only if, $M$ can be decomposed into a product of one- and two-level matrices of type $(*).$
\end{theorem}

Henceforth we shall refer to $(*)$ as the \textbf{generating set for $U_{n}(\D[\omega])$}. 

Central to the proof of Theorem \ref{GSDecompThm} is an exact synthesis algorithm, which accomplishes the requisite decomposition. In principle, this is quite similar to the process of \textit{Gaussian elimination}, which serves to decompose an invertible matrix into \textit{elementary matrices} (corresponding to elementary row operations). In effect, our generators $(*)$ are the `elementary matrices' that are appropriate for unitary operators over $U_n(\D[\omega]).$ Rather than recapitulate the algorithm from \cite{GILES-SEL}, we will detail a (slightly simpler) modified version in the following section. Our algorithm differs only cosmetically from that of Giles and Selinger; we have modified it to simplify subsequent proofs of interest to the thrust of our current work. Consequently, we shall call our altered version the \textit{modified} Giles-Selinger algorithm. In the process, we will also re-prove Theorem~\ref{GSDecompThm}. Indeed, it
is a consequence of Theorem~\ref{AlgTermThm} below.

\section{The Modified Giles-Selinger Algorithm}

The current work will utilise the core of the Giles-Selinger algorithm to drive at our result. Preliminarily, it is of interest to expose some useful and important elements of $\Z[\omega],$ to increase the reader's familiarity with our calculations in following sections.

\begin{example}
\label{RingsExample} The following are elements of $\Z[\omega]$:
\[\begin{array}{ccc}
\sqrt{2}&=&\omega - \omega^3\\
i&=&\omega^2\\
\lambda &=& 1+\sqrt{2}\\
\delta &=& 1+\omega.
\end{array}
\]
Observe that $\omega,$ $i,$ and $\lambda$ are units in the ring, with $\omega^{-1} = \omega^7,$ $i^{-1} = -i,$ and $\lambda^{-1} = \sqrt{2} - 1.$
\end{example}

Recall that $\Z[\omega]$ is a Euclidean domain, and note that $2$ is not prime in this ring. In fact, $2$ factors into 4 (non-distinct) primes, as we have \[
2 = \delta^4\cdot\omega^{-2}\cdot\lambda^{-2}.
\]
Since $\Z[\omega]$ is a ring extension of $\Z$ of degree 4, it follows that $2$ can have at most four prime factors in $\Z[\omega]$; hence it follows that $\delta$ is prime.

The main difference between $\D[\omega]$ and $\Z[\omega]$ is that the former contains dyadic fractions; i.e., fractions with denominators comprising powers of 2. Therefore, every element $u$ of $\D[\omega]$ can be written in the form \[
u = \frac{u'}{2^\ell},
\]
where $u'\in\Z[\omega]$ and $\ell\geq 0$ is an integer.

However, since 2 is not prime, it is often more natural to consider denominators that are powers of $\delta,$ motivating the following

\begin{definition}
\label{LeastDen2Def}
Let $u\in\D[\omega].$ A non-negative integer $k$ is called a \textbf{$\delta$-exponent} for $u$ if $\delta^ku\in\Z[\omega].$ Note that such a $k$ always exists, because $u=\frac{u'}{2^\ell}$ for some $u'\in\Z[\omega]$, hence $\delta^{4\ell}u=u'\omega^{2\ell}\lambda^{2\ell}\in\Z[\omega]$. The smallest such $k$ is called the \textbf{least $\delta$-exponent.} We also write $\lde(u)$ for the least $\delta$-exponent of $u$. 

We also extend the definition of $\delta$-exponent and least $\delta$-exponent to vectors and matrices in the obvious way. In other words, $k$ is a $\delta$-exponent for a vector (or matrix) if it is a $\delta$-exponent for all of its entries, and the least such $k$ is called the least $\delta$-exponent of the vector (or matrix).
\end{definition}

The algorithm we are about to describe seeks to reduce a given matrix to the identity by cutting down its least $\delta$-exponent. First, we make the

\begin{definition}
\label{DeltaResidueDef}
Note that $\Z[\omega]/(\delta) = \{0,1\}$. Namely, since $\omega\equiv 1\pmod{\delta}$, every element of $\Z[\omega]$ is congruent to an integer modulo $\delta$. Moreover, since $2\equiv 0\pmod{\delta}$, every integer is congruent to $0$ or $1$. Now consider the projection map $\rho_\delta\:\Z[\omega]\rightarrow\Z[\omega]/(\delta).$ We call $\rho_\delta$ the \textbf{delta residue map,} and refer to $\rho_\delta(u)$ as the \textbf{delta residue} of $u.$ We can immediately extend $\rho_\delta$ to vectors componentwise: if \[\boldsymbol{v} = \begin{bmatrix}
v_1\\
\vdots\\
v_n
\end{bmatrix},
\mbox{ then }
\rho_\delta(\boldsymbol{v}) = \begin{bmatrix}
v'_1\\
\vdots\\
v'_n
\end{bmatrix},
\]
where $v'_\iota = \rho_{\delta}(v_\iota)\in\{0,1\}.$
\end{definition}

Now we come to a number of useful results which shall allow us to describe the modified Giles-Selinger Algorithm.

\begin{lemma}
\label{OmegaModLemma}
Let $u\in\Z[\omega]$ such that $u\equiv 1\pmod{\delta}.$ Then $u \equiv \omega^m\pmod{\delta^3},$ for some $m\in\{0,\ldots,3\}.$
\end{lemma}
\begin{proof}
We can see this from the fact that \[\Z[\omega]/(\delta^3) = \{1,\omega,\omega^2,\omega^3,0,1+\omega,1+\omega^2,1+\omega^3\}.\]
 Notice that the first four entries are equivalent to $1$ modulo $\delta,$ while the other entries are equivalent to $0$ modulo $\delta.$
\end{proof}
\begin{lemma}
\label{RowReductLemma}
Let 
\[\boldsymbol{u} = \left[\begin{array}{c}
u_1\\
u_2\\
\end{array}\right],\] where $u_1,u_2\in\Z[\omega]$ and  $u_1,u_2\equiv 1\mbox{ (mod }\delta).$ Then there exists $j\in\{0,\ldots,3\}$ such that \[H_{[1,2]}\omega_{[1]}^j\left[\begin{array}{c}
u_1\\
u_2\\
\end{array}\right] = \left[\begin{array}{c}
u'_1\\
u'_2\\
\end{array}\right],\] where $u'_1,u'_2\in\Z[\omega]$ and $u
'_1,u'_2\equiv 0\mbox{ (mod }\delta).$
\end{lemma}
\begin{proof}
Note that, by Lemma \ref{OmegaModLemma}, $u_1\equiv \omega^\ell\pmod{\delta^3}$ and $u_2\equiv\omega^m\pmod{\delta^3}$ for some $\ell,m.$ 
Let $j = m - \ell$. Since $m$ and $\ell$ in Lemma~\ref{OmegaModLemma} can be chosen modulo 4, we can assume without loss of generality that $j\in\{0,\ldots,3\}.$ Then \[\begin{array}{ccc}
H_{[1,2]}\omega_{[1]}^j\left[\begin{array}{c}
u_1\\
u_2\\
\end{array}\right]&=& 
H_{[1,2]}\omega_{[1]}^j\left[\begin{array}{c}
\omega^\ell + \delta^3a\\
\omega^m + \delta^3b\\
\end{array}\right] \\
&=& H_{[1,2]}\left[\begin{array}{c}
\omega^m + \delta^3a\omega^j\\
\omega^m + \delta^3b\\
\end{array}\right]\\
&=&\frac{1}{\sqrt{2}}\left[\begin{array}{c}
2\omega^m + \delta^3(a\omega^j+b)\\
0+\delta^3(a\omega^j-b)
\end{array}
\right]\\
&=&\left[\begin{array}{c}
\delta(\frac{\delta\omega^m}{\lambda\omega}+\lambda\omega(a\omega^j+b))\\
\delta(\lambda\omega(a\omega^j-b))
\end{array}
\right],
\end{array}
\]
as desired.
\end{proof}

\begin{lemma}
\label{ColStepLemma}
Let $\boldsymbol{v}$ be an $n$-dimensional unit vector with entries in $\D[\omega].$ (That is, $\boldsymbol{v}^{\dagger}\boldsymbol{v} = 1.)$ Let $k = \lde(\boldsymbol{v}).$ If $k > 0,$ then there exists a string of generators $G_1,\ldots,G_\ell$ from $(*)$ such that $\lde(G_\ell G_{\ell-1}\ldots G_1\boldsymbol{v}) < k.$
\end{lemma}
\begin{proof}
Let $\boldsymbol{u} = \delta^k\boldsymbol{v}.$ Then the entries of $\boldsymbol{u}$ reside in $\Z[\omega].$ Let $\boldsymbol{u'} = \rho_\delta(\boldsymbol{u}).$ From $\boldsymbol{v}^\dagger\boldsymbol{v} = 1,$ we get $\boldsymbol{u}^\dagger\boldsymbol{u} = (\delta^\dagger\delta)^k\boldsymbol{v}^\dagger\boldsymbol{v} \equiv (\delta^\dagger\delta)^k\equiv 0\pmod{\delta}.$ Hence we see that $u_1'^\dagger u_1' + \ldots +u_n'^\dagger u_n' \equiv 0.$ Thus there are an even number of $\ay\in\{1,\ldots,n\}$ such that $u_\ay' \equiv 1.$

We pair off the nonzero entries in $\boldsymbol{u}$ thus: let $u_{\ay_1}',\ldots,u_{\ay_{2m}}'$ be the non-zero entries of $\boldsymbol{u},$ such that $\ay_1<\ay_2<\ldots<\ay_{2m},$ and we repeatedly apply Lemma \ref{RowReductLemma} to the pairs $u_{\ay_{2\bee-1}}',u_{\ay_{2\bee}}$ for $\bee\in\{1,\ldots,m\}.$ In this way, we find $G_1,\ldots,G_\ell$ such that \[G_\ell\ldots  G_1\left[\begin{array}{c}
u_1\\
\vdots\\
u_n\\
\end{array}\right] = \left[\begin{array}{c}
z_1\\
\vdots\\
z_n\\
\end{array}\right],\] where $z_1,\ldots,z_n\equiv 0\mbox{ (mod }\delta).$ 
Then $\lde(G_\ell\ldots G_1\boldsymbol{v}) < k.$
\end{proof}

\begin{lemma}
\label{SingleEntryLemma}
Let $\boldsymbol{v}$ be a vector with entries in $\Z[\omega]$ such that $\boldsymbol{v}^{\dagger}\boldsymbol{v} = 1.$ Then\[
\boldsymbol{v} = \left[\begin{array}{c}
0\\
\vdots\\
\omega^\ell\\
\vdots\\
0
\end{array} \right],
\]
where $\ell\in\{1,\ldots,8\}.$
\end{lemma}
\begin{proof}
Note that $\boldsymbol{v}^\dagger\boldsymbol{v} = 1 + 0\sqrt{2}.$ Observe that if $x = a\omega^3+b\omega^2+c\omega+d,$ then $x^\dagger x = (a^2+b^2+c^2+d^2)+(ab + bc + cd - ad)\sqrt{2}\in\Z[\sqrt{2}].$ This fact can immediately be generalised to sums of elements in $\Z[\sqrt{2}],$ whence we see that \[
\boldsymbol{v}^\dagger\boldsymbol{v} = v^\dagger_1v_1+\ldots+v^\dagger_nv_n = \sum^n_{j=1}(a_j^2+b_j^2+c_j^2 + d_j^2) + (a_jb_j+b_jc_j+c_jd_j-a_jd_j)\sqrt{2} = 1 + 0\sqrt{2}.
\]
Thus $\sum^n_{j=1}(a_j^2+b_j^2+c_j^2+d_j^2) = 1,$ which implies that precisely one of
\[a_1,\ldots, a_n,b_1,\ldots,b_n,c_1,\ldots,c_n,d_1,\ldots,d_n\] is equal to $\pm 1$, and all other terms are equal to zero. Hence \[
\boldsymbol{v} = \left[\begin{array}{c}
0\\
\vdots\\
\omega^\ell\\
\vdots\\
0
\end{array}\right],
\]for some $\ell\in\{1,\ldots,8\}.$
\end{proof}
\begin{lemma}
\label{ColumnLemma}
Let $\boldsymbol{v}$ be an $n$-dimensional unit vector with entries in $\D[\omega].$ Then there exists a sequence of operators $G_1,\ldots,G_q$ from $(*)$ such that $G_q G_{n-1}\ldots G_1\boldsymbol{v} = e_n$ (i.e., the unit vector with $1$ in the final row and $0$ everywhere else).
\end{lemma}
\begin{proof}
We proceed by induction on $k = \lde(\boldsymbol{v})$. Suppose $k > 0.$ Then by Lemma \ref{ColStepLemma}, there exist generators $G_1,\ldots,G_\ell$ such that $\lde(G_\ell\ldots G_1\boldsymbol{v}) < k,$ whence by the induction hypothesis there exist generators $G_{\ell+1},\ldots,G_q$ such that \[G_q\ldots G_{\ell+1}G_\ell\ldots G_1\boldsymbol{v} = e_n.\]
Now suppose $k = 0.$ By Lemma \ref{SingleEntryLemma}, we see that 
\[ \boldsymbol{v} = \left[\begin{array}{c}
0\\
\vdots\\
\omega^\ell\\
\vdots\\
0
\end{array}\right]
\]
for some $\ell\in\{1,\ldots ,8\}.$  Let $s$ be the index of the non-zero entry. Thus, if $s < n,$ $\omega^{8-\ell}_{[n]}X_{[s,n]}\boldsymbol{v} = e_n.$ If $s = n,$ then $\omega^{8-\ell}_{[n]}\boldsymbol{v} = e_n.$
\end{proof}
The preceding lemmas induce an algorithm which allows us to take an arbitrary element in $U_n(\D[\omega])$ and transform it to the identity operator in a finite number of steps. This is done one column at a time for a given matrix, beginning with the rightmost column and proceeding until it reduces its input to the identity. 

We now give an explicit presentation of the modified Giles-Selinger algorithm.

\begin{algorithm}
\label{GSAlg}
INPUT: A unitary $n\times n$ matrix $U$ with entries in $\D[\omega].$

OUTPUT: A sequence $G_\ell,\ldots ,G_1$ from $(*)$ such that $G_\ell\ldots G_1U = I.$

\begin{enumerate}\setcounter{enumi}{-1}
\item Let $M\leftarrow U$.
\item If $M = I$, stop.
\item Let $\jay$ be the greatest integer such that $Me_{\jay}\neq e_{\jay},$ and let $\boldsymbol{v} = Me_{\jay}.$ Let $k = \lde(\boldsymbol{v})$ and $\boldsymbol{u} = \delta^k\boldsymbol{v}.$
\begin{enumerate}[(a)]
\item CASE $k = 0.$ Then, by Lemma \ref{SingleEntryLemma}, we know that $\boldsymbol{v} = \omega^me_{\ell}$ for some $m\in\{1,\ldots ,8\}$ and $\ell\in\{1,\ldots ,n\}.$ Moreover, $\ell\leq\jay,$ since $M$ is unitary and $Me_{\iota} = e_{\iota}$ for all $\iota>\jay,$ so each entry of $\boldsymbol{v}$ beyond $\jay$ must be equal to zero.
\begin{itemize}
\item If $\ell = \jay,$ let $\vec{G} = \omega^{8-m}_{[\jay]}$.

\item If $\ell\neq\jay,$ let $\vec{G} = \omega^{8-m}_{[\jay]}X_{[\ell,\jay]}.$
\end{itemize}
\item CASE $k > 0:$ Let $\iota,\ell$ be the first two indices such that $u_\iota,u_\ell\equiv 1\pmod{\delta}$ and $\iota < \ell.$ Then, by Lemma \ref{OmegaModLemma}, $u_\iota\equiv\omega^m\pmod{\delta^3}$ and $u_\ell\equiv\omega^q\pmod{\delta^q}$ for some $m,q\in\{0,\ldots ,3\}.$ Let $z = (q - m)\pmod{4},$ and let $\vec{G} = H_{[\iota,\ell]}\omega^{z}_{[\iota]}.$
\end{enumerate}
\item Output $\vec{G}.$ Let $M\leftarrow\vec{G}M.$ Go to step 1.
\end{enumerate}
\end{algorithm}

\begin{definition}
\label{SyllableDef}
Note that step 3 of Algorithm \ref{GSAlg} repeatedly outputs a short sequence of generators of the form $\omega^m_{[\jay]}, \omega^m_{[\jay]}X_{[\ell,\jay]},$ or $H_{[\iota,\jay]}\omega^m_{[\iota]}.$ We refer to these sequences as \textbf{syllables}.
\end{definition}

\begin{theorem}
\label{AlgTermThm}
Algorithm \ref{GSAlg} outputs a finite sequence of syllables $G_1,\ldots ,G_\ell$ such that $G_\ell\ldots G_1 U = I;$ or, equivalently, $G_1^{-1}\ldots G_\ell^{-1} = U.$
\end{theorem}

We shall prove the theorem by induction, which necessitates the following 

\begin{definition}
\label{LevelDef}
Let $M\in U_n(\D[\omega])$. Then the \textbf{level} of $M$ (written $\level(M)$) is a triple $(\jay,\kay,m),$ where:

\begin{itemize}
\item $\jay$ is the greatest index of $\{1,\ldots ,n\}$ such that $Me_{\jay}\neq e_{\jay}.$ If no such index exists, we set $\jay=0;$
\item $\kay$ is the least $\delta$-exponent of $\boldsymbol{v} = Me_{\jay},$ or $0$ if $\jay = 0;$
\item $m$ is the number of entries in $\boldsymbol{v} = Me_{\jay}$ which have least $\delta$-exponent $\kay,$ or $0$ if $\kay = 0.$
\end{itemize}

We say $(\jay,\kay,m) < (\jay',\kay',m')$ if $\jay'<\jay,$ or if $\jay = \jay'$ and $\kay < \kay'$ or if $\jay=\jay'$ and $\kay = \kay'$ and $m < m'$ (i.e., levels have lexicographic order).
\end{definition}

\begin{proof}[Proof of Theorem \ref{AlgTermThm}]
By easy case distinction using Lemmas \ref{ColStepLemma} and \ref{ColumnLemma}, the level of $M$ strictly decreases with each application of steps 2--3. This proves that the algorithm terminates. 

Moreover, after each iteration of the algorithm, we have $M=\vec{G}\cdot U,$ where $\vec{G}$ is the total sequence of syllables output up to that point. Therefore, when the algorithm terminates in step 1, we have $I = \vec{G}\cdot U,$ as claimed.
\end{proof}

\chapter{A Complete Set of Equations for \texorpdfstring{$\boldsymbol{U_4(\D[\omega])}$}{U^^e2^^82^^84(^^e2^^85^^85[^^cf^^89])}}

In the previous chapter, we have shown that every element of $U_n(\D[\omega])$ can be written as a product of generators of the form $(*)$ (see Definition \ref{StarDef}). However, this representation is far from unique; for example, $X_{[1,2]}X_{[2,3]}$ and $X_{[2,3]}X_{[1,3]}$ represent the same operator. The purpose of the current chapter is to give a complete set of relations for these generators for the case $n=4;$ i.e., for the group $U_4(\D[\omega])$.

\section{The Main Result}

As usual, a \textit{word} is a sequence of generators, and two words are \textit{equivalent} if they define the same operator. Of special consideration is the \textit{empty word}, i.e., the word of zero length, denoted $\epsilon.$

\begin{definition}
\label{Equiv1Def}
Let $\vec{G} = G_q\ldots G_1$ be a word (i.e., a sequence of generators from $(*)$). We write $\sem{\vec{G}} = G_q\cdot\ldots \cdot G_1$ for the operator associated with $\vec{G}.$

We say that $\overrightarrow{G}\sim\overrightarrow{H}$ (read: $\overrightarrow{G}$ is \textbf{equivalent} to $\overrightarrow{H}$) if $\sem{\vec{G}} = \sem{\vec{H}}.$
\end{definition}

\begin{definition}
\label{Equiv2Def}
Let $\approx$ be the smallest congruence relation on words satisfying the equations shown in
Table~\ref{tab-equations}.
Here, as usual, a congruence relation is an equivalence relation satisfying \[
\mbox{(cong) }\vec{G}\approx\vec{G'},\vec{H}\approx\vec{H'}\Rightarrow\vec{G}\vec{H}\approx\vec{G'}\vec{H'}.
\]
We say that $\overrightarrow{G}$ is \textbf{equationally equivalent} to $\overrightarrow{H}$ if $\overrightarrow{G}\approx\overrightarrow{H}.$
Note that we do not claim that the equations in Table~\ref{tab-equations} are minimal; it is possible that some of them are consequences of the others.
\end{definition}

\let\mylabel\label

\begin{table}
\[\begin{array}{cccc@{}c}
\myitem\mylabel{ax-a}&\omega^8_{[\jay]}&\approx&\epsilon\\
\myitem\mylabel{ax-g}&H^2_{[\jay,\kay]}&\approx&\epsilon&(\jay<\kay)\\
\myitem\mylabel{ax-i}&X^2_{[\jay,\kay]}&\approx&\epsilon&(\jay<\kay)\\
[2ex]

\myitem\mylabel{ax-w}&\omega_{[\jay]}\omega_{[\kay]}&\approx&\omega_{[\kay]}\omega_{[\jay]}&(\jay\neq\kay)\\
\myitem\mylabel{ax-b}&\omega_{[\ell]}H_{[\jay,\kay]}&\approx&H_{[\jay,\kay]}\omega_{[\ell]} & (\jay<\kay, \ell\neq\jay,\kay)\\
\myitem\mylabel{ax-e}&\omega_{[\ell]}X_{[\jay,\kay]}&\approx&X_{[\jay,\kay]}\omega_{[\ell]}&(\jay<\kay,\ell\neq\jay,\kay)\\
\myitem\mylabel{ax-u}&H_{[\jay,\kay]}H_{[\ell,\tee]}&\approx&H_{[\ell,\tee]}H_{[\jay,\kay]}&(\jay<\kay, \ell<\tee,\{\ell,\tee\}\cap\{\jay,\kay\}=\emptyset)\\
\myitem\mylabel{ax-n}&H_{[\jay,\kay]}X_{[\ell,\tee]}&\approx&X_{[\ell,\tee]}H_{[\jay,\kay]}&(\jay<\kay, \ell<\tee,\{\ell,\tee\}\cap\{\jay,\kay\}=\emptyset)\\
\myitem\mylabel{ax-v}&X_{[\jay,\kay]}X_{[\ell,\tee]}&\approx&X_{[\ell,\tee]}X_{[\jay,\kay]}&(\jay<\kay, \ell<\tee,\{\ell,\tee\}\cap\{\jay,\kay\}=\emptyset)\\
[2ex]

\myitem\mylabel{ax-c}&X_{[\jay,\kay]}\omega_{[\kay]}&\approx&\omega_{[\jay]}X_{[\jay,\kay]}&(\jay<\kay)\\
\myitem\mylabel{ax-d}&X_{[\jay,\kay]}\omega_{[\jay]}&\approx&\omega_{[\kay]}X_{[\jay,\kay]}&(\jay<\kay)\\
\myitem\mylabel{ax-j}&X_{[\jay,\kay]}X_{[\jay,\ell]}&\approx&X_{[\kay,\ell]}X_{[\jay,\kay]}&(\jay<\kay<\ell)\\
\myitem\mylabel{ax-k}&X_{[\jay,\kay]}X_{[\ell,\jay]}&\approx&X_{[\ell,\kay]}X_{[\jay,\kay]}&(\ell<\jay<\kay)\\
\myitem\mylabel{ax-l}&X_{[\jay,\kay]}H_{[\jay,\ell]}&\approx&H_{[\kay,\ell]}X_{[\jay,\kay]}&(\jay<\kay<\ell)\\
\myitem\mylabel{ax-o'}&X_{[\jay,\kay]}H_{[\ell,\jay]}&\approx&H_{[\ell,\kay]}X_{[\jay,\kay]}&(\ell<\jay<\kay)\\
[2ex]

\myitem\mylabel{ax-e'}&\omega_{[\jay]}\omega_{[\kay]}X_{[\jay,\kay]}&\approx&X_{[\jay,\kay]}\omega_{[\jay]}\omega_{[\kay]}&(\jay<\kay)\\
\myitem\mylabel{ax-f'}&\omega_{[\jay]}\omega_{[\kay]}H_{[\jay,\kay]}&\approx&H_{[\jay,\kay]}\omega_{[\jay]}\omega_{[\kay]}&(\jay<\kay)\\
[2ex]

\myitem\mylabel{ax-h}&H_{[\jay,\kay]}X_{[\jay,\kay]}&\approx&\omega^4_{[k]}H_{[\jay,\kay]}&(\jay<\kay)\\
\myitem\mylabel{ax-z4}&H_{[\jay,\kay]}\omega^2_{[\jay]}H_{[\jay,\kay]}&\approx& \omega^6_{[\jay]}H_{[\jay,\kay]} \omega^3_{[\jay]}\omega^5_{[\kay]}&(\jay<\kay)\\
\myitem\mylabel{ax-z5}& H_{[\jay,\kay]}H_{[\ell,\tee]} H_{[\jay,\ell]}H_{[\kay,\tee]}&\approx & H_{[\jay,\ell]}H_{[\kay,\tee]}H_{[\jay,\kay]}H_{[\ell,\tee]} &(\jay<\kay<\ell<\tee)\\
\end{array}
\]
\caption[Equations for the group.]{Equations for $U_n(\D[\omega]),$ where $\jay,\kay,\ell,\tee\in\{1,\ldots ,n\}.$ 
  The equations in the first group give the order of each generator. The
  equations in the second group state that non-overlapping generators commute. 
  The equations in the third group state that $X$ can be used to permute
  indices.
  The equations in the fourth group capture the fact that 
  the matrix
  $\scriptsize\protect\begin{bmatrix}\omega~0\protect\\0~\omega\protect\end{bmatrix}$
  is a scalar matrix, hence commutes with every other $2\times 2$ matrix. Finally, the
  equations in the last group say something non-trivial about the generators.
}
\label{tab-equations}
\end{table}

This brings us to the statement of the most important theorem of this work, which the rest of the thesis is devoted to establishing.

\begin{theorem}
\label{MainThm}
Let $\vec{G}$ and $\vec{H}$ be words over the generating set $(*)$ for $U_4(\D[\omega])$. Then $\vec{G}\approx\vec{H}$ if, and only if, $\vec{G}\sim\vec{H}.$
\end{theorem}

\section{Soundness}
The left-to-right implication of Theorem \ref{MainThm} is called the \textit{soundness} of the equations. It is quite simple, encapsulated by the following

\begin{lemma}[Soundness]
\label{Equiv1Equiv2Lemma1}
Let $\vec{G}$ and $\vec{H}$ be words. Then $\vec{G}\approx\vec{H}$ implies $\vec{G}\sim\vec{H}.$
\end{lemma}
\begin{proof}
Since $\sim$ is a congruence relation, it suffices to show that it satisfies each of the equations in Table~\ref{tab-equations}. This can be verified by direct calculation of each operation.
\end{proof}

\section{Completeness}

The right-to-left implication of Theorem \ref{MainThm} is called \textit{completeness}. Our general approach to proving completeness is as follows: we will define a \textit{normal form}; i.e., a choice of a unique representative of each equivalence class of words. Then we shall exhibit completeness by showing that each word can be converted to normal form using only the given equations. In fact, we will define the normal form of an operator to be precisely the sequence of syllables output by the modified Giles-Selinger algorithm.

Many parts of subsequent proofs work, \textit{mutatis mutandis}, for arbitrary $n;$ there is only one case in the proof of Lemma \ref{TechLemma} which requires specialisation to $n = 4.$ To enable future generalisation, we retain full generality in our methods of proof whenever possible, and specialise only where required. 

We start with a lemma providing some additional equations that will be
useful in the proof below, but are consequences of the equations from
Table~\ref{tab-equations}.

\begin{table}
\[\begin{array}{cccc@{}c}
\myitem\mylabel{ax-p}&H_{[\jay,\ell]}X_{[\jay,\kay]}&\approx&X_{[\jay,\kay]}H_{[\kay,\ell]}&(\jay<\kay<\ell)\\
\myitem\mylabel{ax-o}&H_{[\ell,\jay]}X_{[\jay,\kay]}&\approx&X_{[\jay,\kay]}H_{[\ell,\kay]}&(\ell<\jay<\kay)\\
\myitem\mylabel{ax-m}&X_{[\jay,\kay]}H_{[\jay,\kay]}&\approx&H_{[\jay,\kay]}\omega^4_{[k]}&(\jay<\kay)\\
[2ex]
\myitem\mylabel{ax-f}&\omega^2_{[\jay]}H_{[\jay,\kay]} &\approx& \omega^6_{[\kay]}H_{[\jay,\kay]}\omega^2_{[\kay]}\omega^2_{[\jay]}&(\jay<\kay)\\
\myitem\mylabel{ax-q}&X_{[\jay,\ell]}H_{[\kay,\ell]}&\approx&H_{[\jay,\kay]}\omega^4_{[k]}X_{[\jay,\kay]}X_{[\jay,\ell]}&(\jay<\kay<\ell)\\
\myitem\mylabel{ax-r}&H_{[\jay,\kay]}\omega^2_{[\jay]}H_{[\jay,\kay]}&\approx& X_{[\jay,\kay]} \omega^7_{[\kay]}\omega_{[\jay]}H_{[\jay,\kay]} \omega^2_{[\jay]}&(\jay<\kay)\\
\myitem\mylabel{ax-x}&H_{[\jay,\kay]}\omega_{[\jay]}X_{[\jay,\kay]}&\approx&\omega_{[\jay]}\omega^5_{[\kay]}H_{[\jay,\kay]}\omega^7_{[\jay]}&(\jay <\kay)\\
\myitem\mylabel{ax-z_1}&H_{[\jay,\kay]}\omega^3_{[\jay]}X_{[\jay,\kay]}&\approx&\omega^7_{[\jay]}\omega^3_{[\kay]}X_{[\jay,\kay]}H_{[\jay,\kay]}\omega_{[\jay]}&(\jay <\kay)\\
\myitem\mylabel{ax-z_2}&H_{[\jay,\kay]}\omega^2_{[\jay]}X_{[\jay,\kay]}&\approx&\omega^6_{[\jay]}\omega^2_{[\kay]}X_{[\jay,\kay]}H_{[\jay,\kay]}\omega^2_{[\jay]}&(\jay <\kay)\\
\myitem\mylabel{ax-z_3}&H_{[\jay,\kay]}\omega_{[\jay]}X_{[\jay,\kay]}&\approx&\omega^5_{[\jay]}\omega_{[\kay]}X_{[\jay,\kay]}H_{[\jay,\kay]}\omega^3_{[\jay]}&(\jay <\kay)\\
[2ex]
\myitem\mylabel{ax-s}& H_{[3,4]} H_{[1,2]}X_{[2,3]}H_{[1,2]}H_{[3,4]}&\approx & H_{[1,3]}H_{[2,4]}X_{[2,3]} H_{[2,4]} H_{[1,3]}\\
\myitem\mylabel{ax-t}&H_{[3,4]}H_{[1,2]}X_{[2,3]}H_{[1,2]}H_{[3,4]}&\approx&H_{[2,3]}H_{[1,4]}X_{[2,4]}H_{[1,4]}H_{[2,3]}\\
\end{array}
\]
\caption[Derived equations for the group.]{Derived equations for $U_n(\D[\omega]),$ where $\jay,\kay,\ell,\tee\in\{1,\ldots ,n\}.$}
\label{tab-equations2}
\end{table}

\subsection{Some Useful Equations}
\begin{lemma}\label{lem-derived}
  The equations shown in Table~\ref{tab-equations2} are consequences
  of those in Table~\ref{tab-equations}.
\end{lemma}

\begin{proof}\relax
  Equations \eqref{ax-p}--\eqref{ax-m} are just the inverses of
  \eqref{ax-l}, \eqref{ax-o'}, and \eqref{ax-h}, respectively. Note that
$X$, $H$, and $\omega^4$ are self-inverse by \eqref{ax-a}--\eqref{ax-i}.
  To prove \eqref{ax-f}, we have 
  \[\omega^2_{[\jay]}H_{[\jay,\kay]} 
   \stackrel{\eqref{ax-a}}{\approx}
   \omega^8_{[\kay]}\omega^2_{[\jay]}H_{[\jay,\kay]}
   \stackrel{\eqref{ax-f'}}{\approx}
   \omega^6_{[\kay]}H_{[\jay,\kay]}\omega^2_{[\kay]}\omega^2_{[\jay]}.
  \]
  
  To prove \eqref{ax-q}, we have
  \[\begin{array}{ccc}
    X_{[\jay,\ell]}H_{[\kay,\ell]}
    \stackrel{\eqref{ax-i}}{\approx}
    X_{[\jay,\kay]}X_{[\jay,\kay]}X_{[\jay,\ell]}H_{[\kay,\ell]}
    \stackrel{\eqref{ax-j}}{\approx}
    X_{[\jay,\kay]}X_{[\kay,\ell]}X_{[\jay,\kay]}H_{[\kay,\ell]}
    &\stackrel{\eqref{ax-p}}{\approx}&
    X_{[\jay,\kay]}X_{[\kay,\ell]}H_{[\jay,\ell]}X_{[\jay,\kay]}\\
    &\stackrel{\eqref{ax-o}}{\approx}&
    X_{[\jay,\kay]}H_{[\jay,\kay]}X_{[\kay,\ell]}X_{[\jay,\kay]}\\
    &\stackrel{\eqref{ax-j}}{\approx}&
    X_{[\jay,\kay]}H_{[\jay,\kay]}X_{[\jay,\kay]}X_{[\jay,\ell]}\\
    &\stackrel{\eqref{ax-m}}{\approx}&
    X_{[\jay,\kay]}\omega^4_{[\kay]}H_{[\jay,\kay]}X_{[\jay,\ell]}.\\
  \end{array}
  \]
  
  To prove \eqref{ax-r}, we have \[
\begin{array}{ccc}
 H_{[\jay,\kay]}\omega^2_{[\jay]}H_{[\jay,\kay]}\stackrel{\eqref{ax-z4}}{\approx}\omega^6_{[\jay]}H_{[\jay,\kay]}\omega^3_{[\jay]}\omega^5_{[\kay]}\stackrel{\eqref{ax-f'}\eqref{ax-w}}{\approx}\omega^7_{[\jay]}\omega_{[\kay]}H_{[\jay,\kay]}\omega^4_{[\kay]}\omega^2_{[\jay]}&\stackrel{\eqref{ax-m}}{\approx}&\omega^7_{[\jay]}\omega_{[\kay]}X_{[\jay,\kay]}H_{[\jay,\kay]}\omega^2_{[\jay]}\\
 &\stackrel{\eqref{ax-c}\eqref{ax-d}}{\approx}&X_{[\jay,\kay]} \omega^7_{[\kay]}\omega_{[\jay]}H_{[\jay,\kay]} \omega^2_{[\jay]}.\\
  \end{array}  
\]
  
  To prove \eqref{ax-x}, we have \[H_{[\jay,\kay]}\omega_{[\jay]}X_{[\jay,\kay]}\stackrel{\eqref{ax-c}}{\approx}H_{[\jay,\kay]}X_{[\jay,\kay]}\omega_{[\kay]}\stackrel{\eqref{ax-h}}{\approx}\omega^4_{[\kay]}H_{[\jay,\kay]}\omega_{[\kay]}\stackrel{\eqref{ax-a}\eqref{ax-f'}}{\approx}\omega_{[\jay]}\omega^5_{[\kay]}H_{[\jay,\kay]}\omega^7_{[\jay]}.\]
  
  To prove \eqref{ax-z_1}, we have \[
\begin{array}{ccc}
  H_{[\jay,\kay]}\omega^3_{[\jay]}
  X_{[\jay,\kay]}\stackrel{\eqref{ax-h}
  \eqref{ax-c}}{\approx}\omega^4_{[k]}
  H_{[\jay,\kay]}\omega^3_{[\kay]}
  &\stackrel{\eqref{ax-i}}{\approx}&\omega^4_{[k]}X_{[\jay,\kay]}X_{[\jay,\kay]}
  H_{[\jay,\kay]}\omega^3_{[\kay]}\\
  &\stackrel{\eqref{ax-m}}
  {\approx}&\omega^4_{[\kay]}X_{[\jay,\kay]}
  H_{[\jay,\kay]}\omega^7_{[\kay]}\\
  &\stackrel{\eqref{ax-a}}{\approx}&\omega^4_{[\kay]}X_{[\jay,\kay]}
  H_{[\jay,\kay]}\omega^8_{[\jay]}\omega^7_{[\kay]}\\&\stackrel{\eqref{ax-e'}\eqref{ax-f'}\eqref{ax-a}}{\approx}&
  \omega^7_{[\jay]}\omega^3_{[\kay]}
  X_{[\jay,\kay]}H_{[\jay,\kay]}
  \omega_{[\jay]}.\\  
  \end{array}  
\]
  
  To prove \eqref{ax-z_2}, we have \[
\begin{array}{ccc}
  
  H_{[\jay,\kay]}\omega^2_{[\jay]}X_{[\jay,\kay]}\stackrel{\eqref{ax-c}\eqref{ax-h}}{\approx}\omega^4_{[\kay]}H_{[\jay,\kay]}\omega^2_{[\kay]}&\stackrel{\eqref{ax-i}}{\approx}&\omega^4_{[\kay]}X_{[\jay,\kay]}X_{[\jay,\kay]}H_{[\jay,\kay]}\omega^2_{[\kay]}\\&\stackrel{\eqref{ax-m}}{\approx}&\omega^4_{[\kay]}X_{[\jay,\kay]}H_{[\jay,\kay]}\omega^6_{[\kay]}\\
  &\stackrel{\eqref{ax-a}}{\approx}&\omega^4_{[\kay]}X_{[\jay,\kay]}H_{[\jay,\kay]}\omega^8_{[\jay]}\omega^6_{[\kay]}\\&\stackrel{\eqref{ax-e'}\eqref{ax-f'}\eqref{ax-a}}{\approx}&\omega^6_{[\jay]}\omega^2_{[\kay]}X_{[\jay,\kay]}H_{[\jay,\kay]}\omega^2_{[\jay]}.\\
\end{array}  
\]
  
To prove \eqref{ax-z_3}, we have \[
\begin{array}{ccc}

H_{[\jay,\kay]}\omega_{[\jay]}X_{[\jay,\kay]}\stackrel{\eqref{ax-c}\eqref{ax-h}}{\approx}\omega^4_{[\kay]}H_{[\jay,\kay]}\omega_{[\kay]}&\stackrel{\eqref{ax-i}}{\approx}&\omega^4_{[\kay]}X_{[\jay,\kay]}X_{[\jay,\kay]}H_{[\jay,\kay]}\omega_{[\kay]}\\&\stackrel{\eqref{ax-m}}{\approx}&\omega^4_{[\kay]}X_{[\jay,\kay]}H_{[\jay,\kay]}\omega^5_{[\kay]}\\&\stackrel{\eqref{ax-a}}{\approx}&\omega^4_{[\kay]}X_{[\jay,\kay]}H_{[\jay,\kay]}\omega^8_{[\jay]}\omega^5_{[\kay]}\\&\stackrel{\eqref{ax-e'}\eqref{ax-f'}\eqref{ax-a}}{\approx}&\omega^5_{[\jay]}\omega_{[\kay]}X_{[\jay,\kay]}H_{[\jay,\kay]}\omega^3_{[\jay]}.\\
\end{array}
\]
      
  To prove \eqref{ax-s}, we have
  \[
\begin{array}{ccc}  
  H_{[3,4]} H_{[1,2]}X_{[2,3]}H_{[1,2]}H_{[3,4]}
    &\stackrel{\eqref{ax-o'}\eqref{ax-p}}{\approx}&
    H_{[3,4]} H_{[1,2]}H_{[1,3]}H_{[2,4]}X_{[2,3]}\\
    &\stackrel{\eqref{ax-u}}{\approx}&
    H_{[1,2]} H_{[3,4]}H_{[1,3]}H_{[2,4]}X_{[2,3]}\\
    &\stackrel{\eqref{ax-z5}}{\approx}&
    H_{[1,3]}H_{[2,4]}H_{[1,2]} H_{[3,4]}X_{[2,3]}\\
    &\stackrel{\eqref{ax-o}\eqref{ax-l}}{\approx}&H_{[1,3]}H_{[2,4]}X_{[2,3]}H_{[1,3]} H_{[2,4]}\\
    &\stackrel{\eqref{ax-u}}{\approx}&
    H_{[1,3]}H_{[2,4]}X_{[2,3]}H_{[1,3]} H_{[2,4]}.\\
\end{array}    
    \]
  
  To prove \eqref{ax-t}, we have
  \[
\begin{array}{ccc}
  H_{[3,4]} H_{[1,2]}X_{[2,3]}H_{[1,2]}H_{[3,4]}
    &\stackrel{\eqref{ax-i}}{\approx}&
    H_{[3,4]} H_{[1,2]}X_{[2,3]}H_{[1,2]}H_{[3,4]}X_{[3,4]}X_{[3,4]}\\
    &\stackrel{\eqref{ax-h}}{\approx}&
    H_{[3,4]} H_{[1,2]}X_{[2,3]}H_{[1,2]}\omega^4_{[4]}H_{[3,4]}X_{[3,4]}\\
    &\stackrel{\eqref{ax-b},\eqref{ax-e}}{\approx}&
    H_{[3,4]} \omega^4_{[4]}H_{[1,2]}X_{[2,3]}H_{[1,2]}H_{[3,4]}X_{[3,4]}\\
    &\stackrel{\eqref{ax-m}}{\approx}&
    X_{[3,4]}H_{[3,4]}H_{[1,2]}X_{[2,3]}H_{[1,2]}H_{[3,4]}X_{[3,4]}\\
    &\stackrel{\eqref{ax-s}}{\approx}&
    X_{[3,4]}H_{[1,3]}H_{[2,4]}X_{[2,3]} H_{[2,4]} H_{[1,3]}X_{[3,4]}\\
    &\stackrel{\eqref{ax-u}}{\approx}&
    X_{[3,4]}H_{[2,4]}H_{[1,3]}X_{[2,3]} H_{[1,3]} H_{[2,4]}X_{[3,4]}\\
    &\stackrel{\eqref{ax-o}\eqref{ax-o'}\eqref{ax-k}}{\approx}&
    H_{[2,3]}H_{[1,4]}X_{[2,4]} H_{[1,4]} H_{[2,3]}X_{[3,4]}X_{[3,4]}\\
    &\stackrel{\eqref{ax-i}}{\approx}&
    H_{[2,3]}H_{[1,4]}X_{[2,4]} H_{[1,4]} H_{[2,3]}.\\    
    \end{array}
\]
\end{proof}

\subsection{The State Graph}

It will be useful to visualise operators as edges in a certain graph, which we now define.

\begin{definition}
\label{StateGraphDef}
Fix $n\geq 1.$ The \textbf{state graph} consists of vertices and two different kinds of directed edges.
\begin{itemize}
\item A vertex of the state graph is an element of $U_n(\D[\omega]);$ i.e., it is a unitary $n\times n$ matrix with entries in $\D[\omega].$ We also refer to a vertex as a \textbf{state}.

\item A \textbf{simple edge} is a triple $\langle M', G, M\rangle,$ where $M$ is a state, $G$ is a generator of the form $(*)$ (see Definition \ref{StarDef}), and $M' = \sem{G}M.$ We write $M\edge{G}M'$ for a simple edge.

\item A \textbf{normal edge} is a triple $\langle M', N, M \rangle,$ where $M$ is a state, $N$ is the unique first syllable (see Definition \ref{SyllableDef}) output by Algorithm \ref{GSAlg} on input $M$, and $M' = \sem{N}M.$ We write $M\nedge{N}M'$ for a normal edge.
\end{itemize}
\end{definition}

A few comments on the relationship between states, normal edges, and the overall state graph are warranted.

\begin{remark}
\label{NormUniqueRemark}
For every state $M\neq I,$ there exists a unique normal edge originating at $M.$ Moreover, if $M\nedge{N}M'$ is normal, then $\level(M')<\level(M)$ (as noted in the proof of Lemma \ref{AlgTermThm}). Consequently, there exists a unique sequence $M\nedge{N_1}M_1\nedge{N_2}M_2\nedge{N_3}\ldots \nedge{N_q}I$ of normal edges from $M$ to $I.$ In other words, the subgraph of the state graph consisting of normal edges forms a tree rooted at $I.$ Note that, contrary to the usual convention in graph theory, the normal edges are pointing toward the root, and not away from it.
\end{remark}

Now we establish a convention and a few more definitions which will give us further traction. Note that, for the remainder of this work, we often write states as $s,t,r$ rather than $M,N$ in order to distinguish them from edges. We begin with

\begin{convention}
\label{EdgeConv}
  If $s_0\edge{G_1} s_1\edge{G_2} \ldots \edge{G_n} s_n$ is a sequence of simple edges in the state graph and $\vec{G} = G_n\ldots G_1$, we also write $s_0\edge{\vec{G}}s_n$. (Note that, following the usual convention for matrix multiplication, words are written right-to-left, whereas edges can be written in any direction, but are typically presented left-to-right.) We write $s\edge{\epsilon}s$ for the empty sequence of edges at $s$, or sometimes just $\epsilon$ when $s$ is clear from the context.
\end{convention}

\begin{definition}
\label{SequenceLevelDef}
  Consider a sequence of simple edges $\vec{G} = s_0\edge{G_1}
  s_1\edge{G_2} \ldots \edge{G_n} s_n$. We say that the {\em level}
  of the sequence is the maximum of the levels of the states
  $s_0,\ldots,s_n$; or, in symbols,
  \[ \level(\vec G) = \max \s{\level(s_0),\ldots,\level(s_n)}.
  \]
\end{definition}

\subsection{Basic Edges}

It will be useful to consider a subset of the simple edges, called the \textit{basic} edges, which we define immediately.

\begin{definition}
\label{BasicDef}
Consider the following subset of $(*)$: \[
\{X_{[\alpha,\alpha+1]},H_{[1,2]},\omega_{[1]}\mid 1\leq\alpha < n\}.
\]
We call elements from this set \textbf{basic generators}. Furthermore, an edge $s\edge{G}t$ is \textbf{basic} if $G$ is a basic generator.
\end{definition}

The basic generators generate the same group as the simple ones, as we will show in Lemma~\ref{BasicGenLemma}. First, we need a definition.

\begin{definition}
\label{ExtentDef}
Let $G$ be any generator of type $(*).$ The \textbf{extent} of $G$ is the largest subscript index appearing in $G;$ i.e., $\extent(X_{[\alpha,\beta]})=\beta$, $\extent(H_{[\alpha,\beta]})=\beta,$ and $\extent(\omega_{[\alpha]})=\alpha$ (where it is understood that, where appropriate, $\alpha < \beta.)$ The \textbf{extent of a sequence} $\vec{G} = G_n\ldots G_1$ is just max$\{\extent(G_i)\:1\leq i\leq n\}.$
\end{definition}

\begin{lemma}
\label{BasicGenLemma}
For any simple edge $G$, there exists a sequence $\vec{G}$ of basic edges such that
\begin{enumerate}[(a)]
\item $\vec{G}\approx G.$
\item $\extent(\vec{G})=\extent(G).$
\item $\level(\vec{G})=\level(G).$
\end{enumerate}
\end{lemma}
\begin{proof}
First, observe that, for $\jay < \kay - 1,$ \[X_{[\jay,\kay]} \approx X_{[\kay-1,\kay]}X_{[\jay,\kay-1]}X_{[\kay-1,\kay]},\] which gives a recursive decomposition of $X_{[\jay,\kay]}$ into basic edges.

Observe also that $\omega_{[\iota]} \approx X_{[1,\iota]}\omega_{[1]}X_{[1,\iota]}$ for $\iota>1.$ Note that, for $\kay > 2,$ $H_{[1,\kay]} \approx X_{[2,\kay]}H_{[1,2]}X_{[2,\kay]}.$ Also, for $1 < \jay < \kay,$ we have $H_{[\jay,\kay]}\approx X_{[1,\jay]}H_{[1,\kay]}X_{[1,\jay]}.$ Taken together, these equations provide decompositions satisfying $(a)$ and $(b).$ To prove $(c)$, consider $s\edge{G}t$ and the corresponding $\vec{G} = s\edge{G_1}s_1\edge{G_2}s_2\ldots s_{n-1}\edge{G_n}t$ as defined above.

First, note that $\level(G)\leq\level(\vec{G});$ this is clear, since $\level(\vec{G})$ is the maximum of the levels of $s,s_1,\ldots ,s_{n-1},t,$ which includes $s$ and $t.$ Let $\alpha$ be the extent of $G$, and let $\level(G) = (\jay,\kay,m).$ Observe that $\alpha\leq\jay;$ otherwise, the $\alpha^{th}$ column of both $s$ and $t$ would be $e_\alpha.$ But since $G$ is a generator of extent $\alpha,$ $\sem{G}e_\alpha\neq\alpha.$ This directly contradicts $\sem{G}s=t.$

From $(b),$ it follows that the edges $G_1,\ldots ,G_n$ all have extent at most$\alpha\leq j,$ and therefore can only act on the first $\jay$ columns of their states. It follows that the $`\jay$ parts' of the levels of $s,s_1,\ldots ,s_{n-1},t$ are all $\leq\jay.$

Also, all but one of $G_1,\ldots ,G_n$ are generators of type $X_{[\alpha,\beta]},$ which do not change $\delta$-exponents. It therefore follows that for all $\iota = 1,\ldots ,n-1,$ either $\level(s_\iota)\leq\level(s)$ or $\level(s_\iota)\leq\level(t).$ Hence $\level(\vec{G})\leq\level(G),$ as was to be shown.
\end{proof}

\subsection{The Main Lemma}

The following is our main technical lemma. Because it is very tedious, the proof shall be delayed until the final section of this chapter.

\begin{lemma}[Main Lemma]
\label{TechLemma}
  Let $s,t,r$ be states of dimension $n\leq 4$, $s\nedge{N} t$ a normal edge, and $s\edge{G} r$ a basic edge. Then there exists a state $q$, a sequence of normal edges $r\nedge{\vec{N'}}q$, and a sequence of simple edges $t\edge{\vec{G'}}q$, such that 
  \begin{itemize}
  \item $\vec{N'}G\approx\vec{G'}N$, and 
  \item $\level{\vec{G'}} < \level{s}$.
  \end{itemize}
  In pictures:
    \[ \xymatrix{
    s \ar@{=>}[d]_{N} \ar[r]^{G} &
    r \ar@{:>}[d]^{\vec{N'}} \\
    t \ar@{.>}[r]^{\vec{G'}} &
    q.
  }
  \]
\end{lemma}

\subsection{Proof of Completeness}

This subsection is devoted to proving that Lemma \ref{TechLemma} implies completeness. We start with a slight strengthening of the lemma, by observing that the sequence of edges $\vec{G'}$ can in fact be chosen to be basic. This follows directly from Lemma \ref{BasicGenLemma}. 

Next, we show that a sequence of normal edges, extended by a single basic edge, can be normalised.

\begin{lemma}
\label{NormalCommuteLemma}
Let $s\edge{G}r$ be a basic edge, and let $s\nedge{\vec{N}}I$ be the unique sequence of normal edges from $s$ to $I.$ Let $r\nedge{\vec{M}}I$ be the unique sequence of normal edges from $r$ to $I.$ Then $\vec{M}G\approx\vec{N}.$
\end{lemma}
\begin{proof}
By induction on the level of $s.$ If $\level(s) = (0,0,0),$ then $s = I.$ An easy case distinction shows that the claim holds for all basic generators $G.$

Now suppose $\level(s)\neq(0,0,0),$ and suppose the hypothesis holds for all states of smaller level. In reading the following, it may be helpful to refer to Figure \ref{NormalCommuteFigure}. By assumption, $s\neq I,$ so $\vec{N}$ has at least one term; i.e., $\vec{N} = \vec{N'}N$ (where $\vec{N'}$ may equal $\epsilon$).  By Lemma \ref{TechLemma}, we have\[ \xymatrix{
    s \ar@{=>}[d]_{N} \ar[r]^{G} &
    r \ar@{:>}[d]^{\vec{M'}} \\
    t \ar@{.>}[r]^{\vec{G'}} &
    q,
  }
  \]
  where $\vec{M'}G\approx\vec{G'}N$ and $\level(\vec{G'})<\level(s).$ Also, by Lemma \ref{BasicGenLemma}, we may assume that $\vec{G'}$ consists of basic edges.
  
  Suppose $\vec{G'} = t_0\edge{G_1}t_1\edge{G_2}t_2\ldots t_{\kay-1}\edge{G_{\kay}}t_{\kay},$ where $t_0 = t$ and $t_{\kay} = q.$ For $\iota\in\{0,\ldots ,\kay\},$ let $t_\iota\nedge{\vec{N_\iota}}I$ be the unique sequence of normal edges from $t_\iota$ to $I.$ (Note that $\vec{N_0} = \vec{N'}.)$
  
  Since $\level(t_\iota)<\level(s)$ for all $\iota\in\{0,\ldots ,\kay\},$ by our induction hypothesis, $\vec{N_\iota}G_\iota\approx\vec{N_{\iota-1}}$ for $\iota\in\{1,\ldots ,\kay\}.$ 
  
By congruence, $\vec{N_{\kay}}\vec{G} = \vec{N_{\kay}}G_{\kay}G_{\kay-1}\ldots G_1\approx\vec{N_{\kay-1}}G_{\kay-1}\ldots G_1\approx\ldots \approx\vec{N_0} = \vec{N'}.$ Thus $\vec{N}\approx\vec{N'}N\approx\vec{N_{\kay}}\vec{G'}N\approx\vec{N_{\kay}}\vec{M'}G.$ The last thing to note is that $\vec{N_{\kay}}\vec{M'}\approx\vec{M},$ by uniqueness of normal sequences. 
\end{proof}

\begin{figure}
\begin{center}
\begin{tikzpicture}[scale=2, double distance=3pt] 
  \path (-1,0) node [outer sep=5pt] (state0) {$I$}
        (3,1) node [outer sep=5pt] (state1) {$s$}
        (1.5,1) node [outer sep=5pt] (state3) {$t_0$}
        (1.5,-1) node [outer sep=5pt] (state5) {$t_k$}
	(3,-1) node [outer sep=5pt] (state6) {$r$}
        (1.5,0.33) node [outer sep=5pt] (state7) {$t_1$}
        (1.5,-0.33) node [outer sep=5pt] (state8) {$t_{k-1}$}
  ;
  \draw [double,thick,-implies] (state3) to[bend right=20] node[above=1ex] {$\vec{N_0}=\vec {N'}$\quad} (state0);
  \draw [double,thick,-implies] (state1) to node[above=1ex] {$N$} (state3);
  \draw [double,thick,-implies] (state6) to node[above=1ex] {$\vec{M'}$} (state5);
  \draw [double,thick,-implies] (state5) to[bend left=20] node[above=0.2ex] {$\vec{N_k}$} (state0);
  \draw [double,thick,-implies] (state7) to[bend right=8] node[above=0.2ex] {$\vec{N_1}$} (state0); 
  \draw [double,thick,-implies] (state8) to[bend left=8] node[above=0.2ex] {$\vec{N_{k-1}}$} (state0);
  \draw (0.65,0) node {$\vdots$};
  \draw (1,0) node {IH};
  \draw (1,0.65) node {IH};
  \draw (1,-0.65) node {IH};
  \draw [->,thick] (state1) to[bend left] node[right] {$G$} (state6);
  \draw [->,thick] (state3) to node[right] {$G_1$} (state7);
  \draw (2.4,0) node {Lemma \ref{TechLemma}};
  \draw [->,thick,dotted] (state7) to (state8);
  \draw [->,thick] (state8) to node[right] {$G_k$} (state5);

\end{tikzpicture}
\end{center}
\caption{An illustration of the proof of Lemma \ref{NormalCommuteLemma}.}
\label{NormalCommuteFigure}
\end{figure}
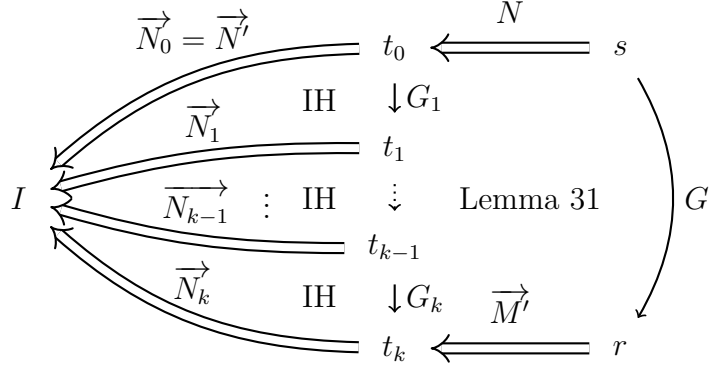

Repeated application of Lemma \ref{NormalCommuteLemma} gives us the following result, which takes us one step closer to completeness:

\begin{lemma}
\label{BasicNormalEquivThm}
Let $s\edge{\vec{G}}I$ be any sequence of basic edges with final state $I,$ and let $s\nedge{\vec{N}}I$ be the unique sequence of normal edges from $s$ to $I.$ Then $\vec{G}\approx\vec{N}.$
\end{lemma}
\begin{proof} (See Figure \ref{BasicNormalEquivFigure} for an illustration of this proof.) By induction on the length of $\vec{G}.$ If $\vec{G}=\epsilon,$ then $s = I$ and we can take $\vec{N}=\epsilon.$ Now let $\vec{G} = s\edge{G}r\edge{\vec{G'}}I.$ Let $r\nedge{\vec{M}}I$ be the unique sequence of normal edges from $r$ to $I,$ and let $s\nedge{\vec{N}}I$ be similar. Then, by our induction hypothesis, $\vec{G'}\approx\vec{M}.$ By Lemma \ref{NormalCommuteLemma}, $\vec{M}G\approx\vec{N}.$ Thus $\vec{G'}\vec{G}\approx\vec{N}.$ 
\end{proof}

\begin{figure}
\begin{center}
\begin{tikzpicture}[scale=2, double distance=3pt]
  \path (-1.5,0) node [outer sep=5pt] (state0) {$I$}
        (5,1) node [outer sep=5pt] (state1) {$s$}
        (1.5,-1) node [outer sep=5pt] (state2) {$\bullet$}
        (3,1) node [outer sep=5pt] (state3) {$\bullet$}
        (1.5,1) node [outer sep=5pt] (state4) {$\bullet$}
        (3,-1) node [outer sep=5pt] (state5) {$\bullet$}
	(4.7,-0.7) node [outer sep=5pt] (state6) {$\bullet$}
	(3,0.3) node [outer sep=5pt] (state7) {$\bullet$}
	(1.5,0.3) node [outer sep=5pt] (state8) {$\bullet$}
	(1.5,-0.1) node [outer sep=5pt] (state9) {$\bullet$}
  ;
  \draw [double,thick,-implies] (state3) to (state4);
  \draw [double,thick,-implies] (state1) to (state3);
  \draw [double,thick,-implies] (state4) to[bend right=18] (state0);
  \draw [double,thick,-implies] (state6) to[bend right=20] (state7);
  \draw [double,thick,-implies] (state7) to (state8);
  \draw [double,thick,-implies] (state8) to[bend right=8] (state0);
  \draw [double,thick,-implies] (state5) to[bend right=8] (state9);
  \draw [double,thick,-implies] (state9) to[bend right=8] (state0);
  \draw [->,thick] (state6) to[bend left=8] (state5);
  \draw [->,thick] (state5) to (state2);
  \draw [->,thick] (state2) to[bend left=20] (state0);
  \draw [->,thick] (state1) to[bend left] (state6);
  \draw [double,thick,-implies] (state2) to[bend right=15] (state0);
  \draw (4.2,0.5) node {Lemma~\ref{NormalCommuteLemma}};
  \draw (2.9,-0.1) node {Lemma~\ref{NormalCommuteLemma}};
  \draw (1.3,-0.4) node {Lemma~\ref{NormalCommuteLemma}};
  \draw (0,-0.6) node {Lemma~\ref{NormalCommuteLemma}};
\end{tikzpicture}
\caption{An illustration of the proof of Lemma \ref{BasicNormalEquivThm}.}
\label{BasicNormalEquivFigure}
\end{center}
\end{figure}

With that in place, we are finally able to prove completeness.

\begin{theorem}
\label{Equiv2Equiv1Thm}
Let $\vec{G},\vec{H}$ be words such that $\vec{G}\sim\vec{H}.$ Then $\vec{G}\approx\vec{H}.$
\end{theorem}
\begin{proof}
By Lemma \ref{BasicGenLemma}, there exist basic words $\vec{G'}$ and $\vec{H'}$ such that $\vec{G}\approx\vec{G'}$ and $\vec{H}\approx\vec{H'}.$ Let $s = \sem{\vec{G'}}^{-1} = \sem{\vec{H'}}^{-1}.$ The second equality is justified since, by Lemma \ref{Equiv1Equiv2Lemma1}, $\vec{G'}\sim\vec{G}\sim\vec{H}\sim\vec{H'}.$ Then $s\edge{\vec{G'}}I$ and $s\edge{\vec{H'}}I.$ Let $s\nedge{\vec{N}}I$ be the unique sequence of normal edges taking $s$ to $I.$  Then, by Lemma \ref{BasicNormalEquivThm}, $\vec{G'}\approx\vec{N}\approx\vec{H'}.$ The result obtains by transitivity.
\end{proof}

Therefore Table \ref{tab-equations} describes a complete set of equations relating generators of the form $(*)$ for the group $U_4(\D[\omega]).$

\subsection{Some Technical Lemmas}

As noted above, we were only able to prove Lemma \ref{TechLemma} in case $n = 4.$ The problematic case arises in Subcase~\ref{bad-case} of the final section of this chapter. The following lemmas are instrumental in dealing with the problematic subcase.

\begin{lemma}
\label{AddCongLemma}
For all $x\in\Z[\omega]$, $\sqrt{2}$ divides $x + x^\dagger$. 
\end{lemma}
\begin{proof}
Let $u = a\omega^3 + b\omega^2 + c\omega + d.$ Then
$u^\dagger = a(\omega^3)^\dagger + b(\omega^2)^\dagger + c \omega^\dagger + d,$
so $u + u^\dagger = a(\omega^3 + (\omega^3)^\dag) + b (\omega^2 + (\omega^2)^\dag) + c (\omega + \omega^\dag) + 2d$. Thus it suffices to demonstrate the result for $\omega+\omega^\dag,$ $\omega^2+(\omega^2)^\dagger,$ and $\omega^3+(\omega^3)^\dag.$

\[
\omega+\omega^\dag = \frac{(1+i)+(1-i)}{\sqrt{2}} = \frac{2}{\sqrt{2}} = \sqrt{2}\mbox{, while }\omega^2+(\omega^2)^\dag = i + (-i) = 0.
\]
Lastly, note that \[
\omega^3+(\omega^3)^\dag = \frac{(i-1)+(-i-1)}{\sqrt{2}} = \frac{-2}{\sqrt{2}} = -\sqrt{2}.
\]
Therefore $x+x^\dagger$ is divisible by $\sqrt{2}.$
\end{proof}

\begin{lemma}
\label{ModDeltaSqLemma}
If $x\equiv y \pmod{\delta^2},$ then $x^\dagger x\equiv y^\dagger y \pmod{2}$.
\end{lemma}

\begin{proof}
 If $x\equiv y\pmod{\delta^2},$ then $x\equiv y\pmod{\sqrt{2}},$ since the moduli only differ by a unit in $\Z[\omega].$ So let $x = y + a\sqrt{2}$. Then $x^\dag x = (y + a\sqrt{2})^\dag(y + a\sqrt{2}) = y^\dag y + (a^\dag y + y^\dag a)\sqrt{2} + 2a^\dag a$. But by Lemma \ref{AddCongLemma}, $\sqrt{2}\divides(a^\dag y + y^\dag a),$ therefore
$2 \divides (a^\dag y + y^\dag a)\sqrt{2},$ and so $ x^\dag x\equiv y^\dag y \pmod{2}$, as claimed.
\end{proof}

\begin{lemma}
\label{RootTwoModTwoLemma}
Let $x\equiv 0$, $y\equiv 1$, $z\equiv \delta$, and $w\equiv \delta+1$ (all modulo $\delta^2$). Then $x^\dagger x + y^\dagger y + z^\dagger z +
w^\dagger w \equiv \sqrt{2}\pmod{2}$. 
\end{lemma}

\begin{proof}
By Lemma \ref{ModDeltaSqLemma}, we have \[\begin{array}{ccc}
x^\dagger x + y^\dagger y + z^\dagger z + w^\dagger w
&\equiv& 0^\dag 0 + 1^\dag 1 + \delta^\dag\delta + (\delta+1)^\dag(\delta+1)\pmod{2}\\
&=& 0 + 1 + \delta^\dag \delta + \delta^\dag\delta + \delta^\dag + \delta + 1\\
&\equiv& \delta^\dag + \delta\pmod{2}\\
&=& 1 + \omega^\dag + 1 + \omega\\
& =& 2 + \sqrt{2}\\
&\equiv& \sqrt{2}\pmod{2}.
\end{array}
\]
\end{proof}

\begin{lemma}
\label{HandWaveLemma}
Let $\boldsymbol{u}$ be a unit vector with entries in $\D[\omega],$ and suppose that $\lde(\boldsymbol{u}) = k.$ Then $k\neq 1.$
\end{lemma}
\begin{proof}
This lemma tests our promise to keep our work self-contained; it is proved by techniques beyond the scope of this thesis. See Case~1 in the proof of Lemma~8.4 of~\cite{ROSS-SEL}.
\end{proof}

\begin{lemma}
\label{NotExist4DimLemma}
There does not exist a 4-dimensional unit vector with least
$\delta$-exponent $k$ and $k$-residues $(0, 1, \delta, \delta+1)\pmod{\delta^2}$.
\end{lemma}

\begin{proof}
Suppose \[\boldsymbol{v} = \frac{1}{\delta^k}\begin{bmatrix}
x\\
y\\
z\\
w
\end{bmatrix}\]
 is such a vector.
 \begin{description}
 \item[{\bf Case 1.}]
  $k = 0.$ By Lemma \ref{SingleEntryLemma}, $\{x,y,z,w\} = \{\omega^n,0,0,0\}$ (as multisets). Thus the residues cannot be as in the hypothesis.
  
  \item[{\bf Case 2.}] $k = 1.$ This residue cannot occur, as noted in Lemma \ref{HandWaveLemma}.

\item[{\bf Case 3.}] $k\geq 2.$ Note that $\boldsymbol{v}$ is a unit vector, and so $x^\dag x + y^\dag y + z^\dag z + w^\dag w = (\delta^\dag\delta)^k.$ Since $k\geq 2$, we have $(\delta^\dag\delta)^k\equiv 0\pmod{2}$. Yet, by Lemma \ref{RootTwoModTwoLemma}, $x^\dag x + y^\dag y + z^\dag z + w^\dag w \equiv \sqrt{2}\pmod{2}$,
a contradiction.\qedhere
  
\end{description}

\end{proof}

\subsection{Proof of the Main Lemma}
For convenience, we restate Lemma \ref{TechLemma}: Let $s,t,r$ be states, $s\nedge{N} t$ a normal edge, and $s\edge{G} r$ a basic edge. Then there exists a state $q$, a sequence of normal edges $r\nedge{\vec{N'}}q$, and a sequence of simple edges $t\edge{\vec{G'}}q$, such that 
  \begin{itemize}
  \item $\vec{N'}G\approx\vec{G'}N$, and 
  \item $\level{\vec{G'}} < \level{s}$.
  \end{itemize}
  In pictures:
    \[ \xymatrix{
    s \ar@{=>}[d]_{N} \ar[r]^{G} &
    r \ar@{:>}[d]^{\vec{N'}} \\
    t \ar@{.>}[r]^{\vec{G'}} &
    q.
  }
  \]

   Note that here we make no restrictions on the dimension of the states involved; this is because, for almost all cases, our proof works in full generality for arbitrary states. We are forced to specialise to states of dimension $4$ only in the troublesome Subcase~\ref{bad-case}, for which the preceding section was developed.
   
   The proof proceeds by case distinction. Note that $t$ and $N$ are uniquely determined by $s$, and $r$ is uniquely determined by $G$. Therefore, it suffices to distinguish cases based on the pair $(s,G)$. Some of the cases will be simplified by the following definition:
  
\begin{definition}
\label{RetroEdgeDef}
Let $s$ and $r$ be states. A basic edge $s\edge{G}r$ is \textbf{retrograde} if there is a normal edge $r\nedge{N}s$ such that $NG\approx\epsilon.$
\end{definition}
    
  Observe that if $s\edge{G}r$ is retrograde, then the result easily follows; the modified Giles-Selinger algorithm produces $r\nedge{\vec{N'}}s\nedge{N}t,$ where $\vec{N'}G\approx\epsilon.$ Thus, in these instances, taking $q = t$ and $\vec{G'}\approx\epsilon$ will satisfy the statement of the theorem. In pictures: \[ \xymatrix{
    s \ar@{=>}[dd]_{N}
     \ar[r]^{G} &
    r \ar@{=>}[d]^{\vec{N'}} \\
    &s\ar@{=>}[d]^{N} \\
    t \ar@{.>}[r]^{\epsilon} &
    t.
  }
  \]
 We shall make frequent use of this fact in many of the subcases to come. Before we proceed, it will be convenient to establish a few more conventions, to help the proof's production and readability.

\begin{convention}
\label{ModConv}
  We schematically write $x\,(\delta^n)$ to denote an arbitrary
  element $y\in\Z[\omega]$ such that $y\equiv x\pmod{\delta^n}$, or in
  other words, an arbitrary number of the form $y=x+a\delta^n$, where
  $a\in\Z[\omega]$. For example, we write
  \[ v = \frac{1}{\delta^k}\begin{bmatrix}
    1\,(\delta^3)\\
    1\,(\delta^3)\\
    0\,(\delta)\\
    0\,(\delta)\\
  \end{bmatrix}
  \]
  \vspace{-1.5em}
  \par\noindent
  to denote any vector of the form
  \vspace{-0.5em}
  \[ v = \frac{1}{\delta^k}\begin{bmatrix}
    1 + a\delta^3\\
    1 + b\delta^3\\
    c\delta\\
    d\delta\\
  \end{bmatrix},
  \]
  where $a,b,c,d\in\Z[\omega]$. This convention is useful in
  organising some of the case distinctions that follow.
\end{convention}

\begin{convention}
\label{OmegaConv}
  We globally assume the relation $\omega^8_{[j]}\approx\epsilon$; i.e., we treat powers of $\omega$ as if they were taken modulo 8 without further mention.
\end{convention}

In what follows, we say that the {\em pivot column} of a unitary $n\times n$ matrix $s$ is the rightmost column where $s$ differs from the identity
  matrix.

  \begin{description}
  \case\label{case-1} $G=H_{[1,2]}$.

    \begin{description}
    \subcase The pivot column of $s$ is of the form
      \[ \boldsymbol{v} = \frac{1}{\delta^k}\begin{bmatrix}  
        \omega^{\iota}\,(\delta^3)\\ \omega^{\iota}\,(\delta^3)\\ 
        x\\
        \vdots\\
        z
      \end{bmatrix}.
      \]
      In this case, the modified Giles-Selinger algorithm specifies
      the syllable $N=H_{[1,2]}$, hence $t=r$ and we can choose $\vec{N'}=\epsilon$ and $\vec{G'}=\epsilon$:
      \[ \xymatrix{
        s \ar@{=>}[d]_{H_{[1,2]}} \ar[r]^{H_{[1,2]}} &
        t \ar@{=>}[d]^{\epsilon} \\
        t \ar@{->}[r]^{\epsilon} &
        t
      }
      \]
      Since $H_{[1,2]}$ is a syllable, $\level(t) < \level(s)$.

    \subcase\label{subcase-1.2} The pivot column of $s$ is of the form
      \[ \boldsymbol{v} = \frac{1}{\delta^k}\begin{bmatrix} 
        \omega^\iota\,(\delta^3)\\ \omega^{\iota+1}\,(\delta^3)\\ x\\\vdots\\z
      \end{bmatrix}.
      \]
      In this case, the modified Giles-Selinger algorithm specifies the syllable $s\nedge{H_{[1,2]}\omega_{[1]}}t$. The state
      $r=H_{[1,2]}s$ has pivot column of the form
      \[ H_{[1,2]}\boldsymbol{v} 
      = H_{[1,2]}\frac{1}{\delta^k}\begin{bmatrix} \omega^\iota + a\delta^3\\
      \omega^{\iota+1} +b\delta^3\\
        \vdots
      \end{bmatrix}
      = \frac{1}{\delta^{k}\sqrt{2}}\begin{bmatrix}
        \omega^\iota+\omega^{\iota+1}+(a+b)\delta^3 \\ \omega^\iota-\omega^{\iota+1}+(a-b)\delta^3\\
        \vdots
      \end{bmatrix}
      \]
      \[
      = \frac{1}{\delta^{k+2}}\begin{bmatrix}
        \lambda(\omega^{\iota+1}+\omega^{\iota+2})+\lambda\omega(a+b)\delta^3 \\ \lambda(\omega^{\iota+1}-\omega^{\iota+2})+\lambda\omega(a-b)\delta^)\\
        \vdots
      \end{bmatrix}
      = \frac{1}{\delta^{k+1}}\begin{bmatrix} 
        \omega^m\,(\delta^3)\\
        \omega^m\,(\delta^3)\\
        \vdots
      \end{bmatrix}.
      \]
      To justify the preceding calculation, we note that $\sqrt{2} = \delta^2\lambda^{-1}\omega^{-1}$ and
      $\omega^{\iota+1}\pm\omega^{\iota+2}\equiv\delta\mmod{\delta^3}$. 
            
      The modified Giles-Selinger algorithm specifies the syllable
      $r\nedge{H_{[1,2]}}s$, and hence $s\edge{H_{[1,2]}}s$ is retrograde.
      
\subcase
      The pivot column of $s$ is of the form
       \[ \boldsymbol{v} = \frac{1}{\delta^k}\begin{bmatrix} 
        \omega^{\iota}\,(\delta^3)\\ \omega^{\iota+3}\,(\delta^3)\\ x\\ \vdots\\z
      \end{bmatrix}. 
      \]
	This case is very similar to Subcase~\ref{subcase-1.2}.
	      
    \subcase The pivot column of $s$ is of the form
      \[ \boldsymbol{v} = \frac{1}{\delta^k}\begin{bmatrix} 
        \omega^\iota\,(\delta^3)\\ \omega^{\iota+2}\,(\delta^3)\\ x\\\vdots\\z
      \end{bmatrix}.
      \]
      In this case, the modified Giles-Selinger algorithm specifies
      the syllable
      $s\nedge{H_{[1,2]}\omega^2_{[1]}}t.$ The state
      $r=H_{[1,2]}s$ has pivot column of the form
      \[ H_{[1,2]}\boldsymbol{v} 
      = H_{[1,2]}\frac{\omega^\iota}{\delta^k}\begin{bmatrix} 1 + a\delta^3\\
        \omega^2+b\delta^3\\         
        \vdots
      \end{bmatrix}
      = \frac{\omega^\iota}{\delta^{k}\sqrt{2}}\begin{bmatrix}
        1+\omega^2+(a+b)\delta^3 \\ 1-\omega^2+(a-b)\delta^3\\
\vdots
\end{bmatrix}
      \]
      \[
      = \frac{\lambda\omega^{\iota+1}}{\delta^{k+2}}\begin{bmatrix}
        \delta^2+(a+b)\delta^3-2\omega \\ 
        \delta^2+(a+b)\delta^3-2\omega-2\omega^2-2b\delta^3\\
\vdots
      \end{bmatrix}
      \]
      \[
      = \frac{\lambda\omega^{\iota+1}}{\delta^{k}}\begin{bmatrix}
        1+(a+b)\delta-2\omega/\delta^2 \\ 
        1+(a+b)\delta-2\omega/\delta^2-2\omega^2/\delta^2-2b\delta\\
        \vdots
      \end{bmatrix}
      \]
      \[
      = \frac{1}{\delta^{k+1}}\begin{bmatrix} 
        \omega^m\,(\delta^3)\\
        \omega^m+\delta^2\,(\delta^3)\\\vdots
      \end{bmatrix}
      = \frac{1}{\delta^{k+1}}\begin{bmatrix} 
        \omega^m\,(\delta^3)\\
        \omega^{m+2}\,(\delta^3)\\ \vdots\end{bmatrix}.
      \]
      To justify the preceding calculation, we note that
      $\lambda\omega^{\iota+1}(1+(a+b)\delta-2\omega/\delta^2)$ is an element of
      $\Z[\omega]$ that is not divisible by $\delta$, and hence is
      congruent to $\omega^m\mmod{\delta^3}$, for some $m$. Moreover,
      $-2\omega^2/\delta^2\equiv \delta^2\mmod{\delta^3}$. Finally, it
      is easy to check that for all $m$, we have
      $\omega^m+\delta^2\equiv \omega^{m+2}\mmod{\delta^3}$.

      The modified Giles-Selinger algorithm specifies the syllable
      $r\nedge{H_{[1,2]}\omega^2_{[1]}}q,$ and so we have $\vec{N'}=H_{[1,2]}\omega^2_{[1]}$. Thus, by equation \eqref{ax-r} from Table \ref{tab-equations2}, we can choose $\vec{G'}=X_{[1,2]} \omega^7_{[2]}\omega_{[1]}$:
      \[ \xymatrix{
        s \ar@{=>}[d]_{H_{[1,2]}\omega^2_{[1]}} \ar[rr]^{H_{[1,2]}} &&
        r \ar@{=>}[d]^{H_{[1,2]}\omega^2_{[1]}} \\
        t \ar[rr]^{X_{[1,2]} \omega^7_{[2]}\omega_{[1]}} &&
        q
      }
      \]
Since $H_{[1,2]}\omega^2_{[1]}$ is a syllable, we have $\level(t)<\level(s).$ Moreover, neither $X_{[1,2]}$ nor the $\omega$ gates alter a state's least $\delta$-exponent, and so we must have that $\level(\vec{G')}<\level(s)$.

    \subcase The pivot column of $s$ is of the form
      \[ \boldsymbol{v} = \frac{1}{\delta^k}\begin{bmatrix}  
        \omega^\iota\,(\delta^3)\\
         0\,(\delta)\\
          x\\
          \vdots\\
          z
      \end{bmatrix}.
      \]
      In this case, the state $r=H_{[1,2]}s$ has pivot
      column 
      \[ H_{[1,2]}\boldsymbol{v} = \frac{1}{\delta^{k+2}}\begin{bmatrix}
        \omega^{m}\,(\delta^3)\\ \omega^{m}\,(\delta^3)\\
        \vdots\\
      \end{bmatrix},
      \]
      and the modified Giles-Selinger algorithm specifies the syllable
      $r\nedge{H_{[1,2]}}s.$ By the equation $H_{[1,2]}^2 \approx \epsilon$ (equation $\eqref{ax-g}$ in Table \ref{tab-equations}), we see that $s\edge{H_{[1,2]}}r$ is retrograde, and the result obtains.
      (Note: the case 
      \[ \boldsymbol{v} = \frac{1}{\delta^k}\begin{bmatrix} 
        0\,(\delta)\\\omega^{\iota}\,(\delta^3)\\ x\\ \vdots\\ z
      \end{bmatrix}
      \] is also retrograde for precisely the same reason.)
      
    \subcase The pivot column of $s$ is of the form
      \[ \boldsymbol{v} = \frac{1}{\delta^k}\begin{bmatrix}  
        0\,(\delta)\\ 0\,(\delta)\\
        \vdots\\
        \omega^{\iota}\,(\delta^3)\\ 
		\vdots\\
		0\,(\delta)\\
        \vdots\\        
        \omega^{h}\,(\delta^3)\\
      	x\\
      	\vdots\\
      	z
      \end{bmatrix}.
      \]
    Let $\alpha$ and $\beta$ be the first two indices of $\delta^k\boldsymbol{v}$ which are not equal to $0\pmod{\delta}.$ The modified Giles-Selinger algorithm specifies the syllable $N = s\nedge{H_{[\alpha,\beta]}\omega^{h-\iota}_{[\alpha]}}t.$ The state $r = H_{[1,2]}s$ has pivot column of the form
    
    \[ H_{[1,2]}\boldsymbol{v} = \frac{1}{\delta^k}
    \begin{bmatrix}
    (a+b)\frac{\delta}{\sqrt{2}}\\
    (a-b)\frac{\delta}{\sqrt{2}}\\
    \vdots\\
    \omega^\iota\,(\delta^3)\\
    \vdots\\
    \omega^h\,(\delta^3)\\
    \vdots
    \end{bmatrix}
    ,
    \]
    where $a,b\in\Z[\omega],$ and the least $\delta$-exponent and the modified Giles-Selinger algorithm's action depend upon the values of $a\pm b.$ Thus we must here refine our case distinction.
    
    \begin{description}
    \subsubcase $(a\pm b)\equiv0\pmod{\delta^2}.$ Note $\delta^3\divides (a\pm b)\delta,$ and so $(a\pm b)\frac{\delta}{\sqrt{2}}\equiv0\pmod{\delta}$. Thus the modified Giles-Selinger algorithm prescribes the syllable $r\nedge{H_{[\alpha,\beta]}\omega^{h-\iota}_{[\alpha]}}q,$ so we can take $\vec{N'} = H_{[\alpha,\beta]}\omega^{h-\iota}_{[\alpha]}$ and $\vec{G'} = H_{[1,2]}:$ 
    \[ \xymatrix{
        s \ar@{=>}[d]_{H_{[\alpha,\beta]\omega^{h-\iota}_{[\alpha]}}} \ar[r]^{H_{[1,2]}} &
        r \ar@{=>}[d]^{H_{[\alpha,\beta]\omega^{h-\iota}_{[\alpha]}}} \\
        t \ar@{->}[r]^{H_{[1,2]}} &
        q
      }
      \]
    By direct calculation, we see that $H_{[1,2]}$ does not increase levels, so $\level(s) = \level(r).$ Also, since $H_{[\alpha,\beta]}\omega^{h-\iota}_{[\alpha]}$ is a syllable, $\level(t)<\level(s)$ and $\level(q)<\level(r).$ Hence $\level(t\edge{H_{[1,2]}}q) < \level(s).$ The equation utilised is $H_{[1,2]}H_{[\alpha,\beta]}\omega_{[\alpha]} \approx H_{[\alpha,\beta]}\omega_{[\alpha]}H_{[1,2]}$ (when $2 < \alpha<\beta,$ as is the case here), which follows from equations $\eqref{ax-b}$ and $\eqref{ax-k}$ in Table \ref{tab-equations}.
    
    \subsubcase $a\pm b\equiv\delta\pmod{\delta^2}.$ 
    
    Then $(a+b)\frac{\delta}{\sqrt{2}}\equiv 1\pmod{\delta}$ and $(a+b)\frac{\delta}{\sqrt{2}} \equiv (a - b)\frac{\delta}{\sqrt{2}}\pmod{\delta^3},$ whence 
    \[
	H_{[1,2]}\boldsymbol{v} = \frac{1}{\delta^k}\begin{bmatrix}
	\omega^\ell\,(\delta^3)\\
	\omega^\ell\,(\delta^3)\\
	\vdots\\
	\omega^\iota\,(\delta^3)\\
	\vdots\\
	\omega^h\,(\delta^3)\\
	\vdots
\end{bmatrix}.	    
    \]
    The modified Giles-Selinger algorithm specifies the syllable $r\nedge{H_{[1,2]}}s$, whence $s\edge{H_{[1,2]}}r$ is retrograde.       
      \subsubcase $(a\pm b) \equiv 1\pmod{\delta}.$ Then \[
		H_{[1,2]}\boldsymbol{v} = \frac{1}{\delta^{k+1}}      \begin{bmatrix}
	\omega^\ell\,(\delta^3)\\
	\omega^\ell\,(\delta^3)\\
	\vdots\\
	0\,(\delta)\\
	\vdots\\
	0\,(\delta)
\end{bmatrix}.
      \]
      Thus the modified Giles-Selinger algorithm prescribes the syllable $r\nedge{H_{[1,2]}}s.$ Thus $s\edge{H_{[1,2]}}r$ is (again) retrograde, and so the result holds for this case.       
    \end{description}
  
  \subcase\label{subcase-1.7} The pivot column of $s$ is of the form \[ \boldsymbol{v} = \begin{bmatrix}  
		\omega^{\iota}\\ 
		0\\
		\vdots\\
		0
      \end{bmatrix}.
      \]
      There are two subcases to consider.
      \begin{description}
      \subsubcase $j = 1.$ Then the \textit{second} column of $s$ is \[
	\boldsymbol{v'} = 
	\begin{bmatrix}
	0\\
	1\\
	0\\
	\vdots
\end{bmatrix},	      
      \]
      and the second column (necessarily the pivot column) of $r$ is \[
	H_{[1,2]}\boldsymbol{v'} = \begin{bmatrix}
	1/\sqrt{2}\\
	-1/\sqrt{2}\\
	0\\
	\vdots
\end{bmatrix}.	      
      \]
      Thus the modified Giles-Selinger algorithm specifies the syllable $r\nedge{H_{[1,2]}}s,$ whence $s\edge{H_{[1,2]}}r$ is retrograde.
      
      \subsubcase $j\geq 2.$ Then the pivot column of $r = H_{[1,2]}s$ is \[
H_{[1,2]}\boldsymbol{v} = \begin{bmatrix}  
		\sqrt{2}\omega^{\iota}\\ 
		\sqrt{2}\omega^{\iota}\\ 
		0\\
		\vdots\\
		0
      \end{bmatrix} = \begin{bmatrix}
      \omega^{\iota-1}\delta^2/\lambda\\
      \omega^{\iota-1}\delta^2/\lambda\\
      0\\
      \vdots\\
      0
      \end{bmatrix} = \frac{1}{\delta^2}\begin{bmatrix}
      \omega^m\\
      \omega^m\\
      0\\
      \vdots\\
      0
      \end{bmatrix}.
      \]
      Thus the modified Giles-Selinger algorithm specifies the syllable $r\nedge{H_{[1,2]}}s,$ and therefore $s\edge{H_{[1,2]}}r$ is retrograde (though for a different reason to the above). 
      \end{description}
      
\subcase The pivot column of $s$ is of the form \[\boldsymbol{v} = \begin{bmatrix}  
		0\\
		\omega^{\iota}\\ 
		0\\
		\vdots\\
		0
      \end{bmatrix}.\]
     This is virtually identical to Subcase~\ref{subcase-1.7}.      

      \subcase The pivot column of $s$ is of the form \[
      \boldsymbol{v} = \begin{bmatrix}  
		0\\
		0\\
		\vdots\\
		\omega^{\iota}\\
		\vdots\\
		0
		\end{bmatrix}.
      \]
      Let $\alpha$ be the index of the non-zero row of $\boldsymbol{v},$ and note that $\alpha\leq\jay$ (the index of the pivot column), since the state is unitary. Then the modified Giles-Selinger algorithm prescribes the syllable $s\nedge{\omega^{8-\iota}_{[\jay]}X_{[\alpha,\jay]}}t.$  Observe that the pivot column of $r = H_{[1,2]}s$ is simply $\boldsymbol{v}$ once again. Thus the modified Giles-Selinger algorithm supplies $\vec{N'} = \vec{N} = \omega^{8-\iota}_{[\jay]}X_{[\alpha,\jay]}.$ (Here we make the convention that $X_{[\jay,\jay]}=\epsilon,$ to cover the case $\alpha=\jay$.) Then, following equations $\eqref{ax-b}$ and $\eqref{ax-n}$ from Table \ref{tab-equations} (which affirm that operators with non-intersecting indices commute), we can set $\vec{G'} = H_{[1,2]}:$\[ \xymatrix{
        s \ar@{=>}[d]_{\omega^{8-\iota}_{[\jay]}X_{[\alpha,\jay]}} \ar[r]^{H_{[1,2]}} &
        r \ar@{=>}[d]^{\omega^{8-\iota}_{[\jay]}X_{[\alpha,\jay]}} \\
        t \ar@{->}[r]^{H_{[1,2]}} &
        q
      }
      \]
      With this we have completely described Case~\ref{case-1}.     
  \end{description}  
\vspace{-3ex}
  
\case\label{case-2} $G = \omega_{[1]}.$
	\begin{description}
          \vspace{-1ex}
  		\subcase The pivot column of $s$ is of the form \[ \boldsymbol{v} = \frac{1}{\delta^k}\begin{bmatrix}  
         \omega^{\iota}\,(\delta^3)\\
         0\,(\delta)\\\vdots\\ \omega^{h}\,(\delta^3)\\x\\\vdots\\z
      \end{bmatrix},
      \]
    where 1 and $\alpha$ are the indices of the first two non-zero rows modulo $\delta,$ and $\alpha\geq 2.$ The modified Giles-Selinger algorithm instructs us to take the syllable $N = H_{[1,\alpha]}\omega^{h-\iota}_{[1]}.$ The pivot column of $r = \omega_{[1]}s$ is \[
	\omega_{[1]}\boldsymbol{v} = \frac{1}{\delta^k}\begin{bmatrix}  
         \omega^{\iota+1}\,(\delta^3)\\
         0\,(\delta)\\
         \vdots\\
          \omega^{h}\,(\delta^3)\\\vdots
      \end{bmatrix}.  	
  	\]	
  	The modified Giles-Selinger algorithm gives reduction step $N' = r\nedge{H_{[1,\alpha]}\omega^{h-\iota-1}_{[1]}}t,$ or $r\nedge{H_{[1,\alpha]}\omega^{h-\iota+3}_{[1]}}t,$ depending on whether $h-\iota\equiv 0\pmod{4}$ or not. 

In case $N' = H_{[1,\alpha]}\omega^{h-\iota-1}_{[1]},$ $s\nedge{\omega_{[1]}}r$ is `almost' retrograde. More precisely, we have:\[
\xymatrix{
s\ar@{=>}[d]_{H_{[1,\alpha]}\omega^{h-\iota}}
\ar@{->}[r]^{\omega_{[1]}} & 
r \ar@{=>}[d]^{H_{[1,\alpha]}\omega^{h-\iota-1}_{[1]}}\\
t\ar@{->}[r]^{\epsilon} & 
t}  	
  	\]
  	Thus the statement of the theorem holds for the same reasons as in retrograde cases.
  	
  	In case $N' = H_{[1,\alpha]}\omega^{h-\iota+3}_{[1]},$ we use the equation ($\eqref{ax-m}$ from Table \ref{tab-equations2}) $H_{[1,\alpha]}\omega^4_{[1]}\approx X_{[1,\alpha]}H_{[1,\alpha]}$ to get:\[
\xymatrix{
s\ar@{=>}[d]_{H_{[1,\alpha]}\omega^{h-\iota}_{[1]}}
\ar@{->}[r]^{\omega_{[1]}} & 
r \ar@{=>}[d]^{H_{[1,\alpha]}\omega^{h-\iota+3}_{[1]}}\\
t\ar@{->}[r]^{X_{[1,\alpha]}} & 
q.
}  	
  	\]
	
   	\subcase The pivot column of $s$ is of the form \[ \boldsymbol{v} = \frac{1}{\delta^k}\begin{bmatrix}  
		0\,(\delta)\\
		\vdots\\
		\omega^{\iota}\,(\delta^3)\\ 
		0\,(\delta)\\
		\vdots\\
		\omega^{h}\,(\delta^3)\\ 
		x\\
		\vdots\\
		z
      \end{bmatrix}.
      \]
    In this case, the pivot column of $r = \omega_{[1]}s$ is of the same form, too. Thus the modified Giles-Selinger algorithm prescribes the same syllable $N = N' = H_{[\alpha,\beta]}\omega^{h-\iota}_{[\alpha]}$ in the states $s$ and $r.$ Equations \eqref{ax-w} and \eqref{ax-b} from Table \ref{tab-equations} allow us to take $\vec{G'} = \omega_{[1]}$:\[ \xymatrix{
        s \ar@{=>}[d]_{H_{[\alpha,\beta]}\omega^{h-\iota}_{[\alpha]}} \ar[r]^{\omega_{[1]}} &
        r \ar@{=>}[d]^{H_{[\alpha,\beta]}\omega^{h-\iota}_{[\alpha]}} \\
        t \ar@{->}[r]^{\omega_{[1]}} &
        q
      }
      \]
      Since $\omega_{[1]}$ does not affect the state's level and $H_{[\alpha,\beta]}\omega^{h-\iota}_{[\alpha]}$ is a syllable, we must have $\level(t) = \level(q) < \level(s).$ 
   	\end{description}
\subcase The pivot column of $s$ is of the form \[     
      \boldsymbol{v} = \begin{bmatrix}  
		\omega^{\iota}\\
		0\\
		\vdots\\
		0
		\end{bmatrix}.
      \]
Then the modified Giles-Selinger algorithm specifies the syllable $s\nedge{\omega^{8-\iota}_{[\jay]}X_{[1,\jay]}}t,$ while the pivot column of $r = \omega_{[1]}s$ is \[
      \boldsymbol{v} = \begin{bmatrix}  
		\omega^{\iota+1}\\
		0\\
		\vdots\\
		0
		\end{bmatrix}.
\]
Thus the modified Giles-Selinger algorithm prescribes reduction step $r\nedge{\omega^{8-\iota-1}_{[\jay]}X_{[1,\jay]}}t.$ (Note that we have here again made the convention that $X_{[\jay,\jay]}=\epsilon,$ to condense the presentation of the case distinction.) Schematically:\[ \xymatrix{
        s \ar@{=>}[d]_{\omega^{8-\iota}_{[1]}X_{[1,\jay]}} \ar[r]^{\omega_{[1]}} &
        r \ar@{=>}[d]^{\omega^{8-\iota-1}_{[1]}X_{[1,\jay]}} \\
        t \ar@{->}[r]^{\epsilon} &
        t.
      }
      \]

\subcase The pivot column of $s$ is of the form \[      
v = \begin{bmatrix}  
		0\\
		\vdots\\
		\omega^{\iota}\\
		0\\
		\vdots\\
		0
		\end{bmatrix}.
\]
Let $\jay$ be the index of the pivot column, and let $\alpha>1$ be the index of the non-zero entry of $\boldsymbol{v}.$ Again making the convention that $X_{[\jay,\jay]}=\epsilon.$ Then the modified Giles-Selinger algorithm prescribes reduction step $s\nedge{\omega^{8-\iota}_{[\jay]}X_{[\alpha,\jay]}}t,$ while the pivot column of $r = \omega_{[1]}s$ is just $\boldsymbol{v}$ once again, whence the modified Giles-Selinger algorithm specifies the reduction step $r\nedge{\omega^{8-\iota}_{[\jay]}X_{[\alpha,\jay]}}t.$ Thus we have:\[ \xymatrix{
        s \ar@{=>}[d]_{\omega^{8-\iota}_{[1]}X_{[\alpha,\jay]}} \ar[r]^{\omega_{[1]}} &
        r \ar@{=>}[d]^{\omega^{8-\iota}_{[1]}X_{[\alpha,\jay]}} \\
        t \ar@{->}[r]^{\omega_{[1]}} &
        q.
      }
      \]   	
      The correctness of this is ensured by equations $\eqref{ax-w}$ and $\eqref{ax-c}$ from Table \ref{tab-equations}. Moreover, $\level(t) = \level(q) < \level(s).$ We have thus exhausted all subcases for Case~\ref{case-2}.

  \case\label{case-3} $G=X_{[1,2]}$.

    \begin{description}
    \subcase The pivot column of $s$ is of the form
      \[ \boldsymbol{v} = \frac{1}{\delta^k}\begin{bmatrix}  
        \omega^{\iota}\,(\delta^3)\\ 
        \omega^{h}\,(\delta^3)\\  
        x\\
        \vdots\\
        z
      \end{bmatrix},
      \]
    Note that the pivot column of $r = X_{[1,2]}s$ is then \[
    \boldsymbol{v} = \frac{1}{\delta^k}\begin{bmatrix}  
        \omega^{h}\,(\delta^3)\\  
        \omega^{\iota}\,(\delta^3)\\ 
        x\\
        \vdots\\
        z
      \end{bmatrix}.
    \] 
    \newpage 
    The modified Giles-Selinger algorithm specifies syllables $N = H_{[1,2]}\omega^{h-\iota\pmod{4}}_{[1]}$ and $\vec{N'} = H_{[1,2]}\omega^{\iota-h\pmod{4}}_{[1]}.$ There are four scenarios, depending on $h-\iota\pmod{4}$:\[ \xymatrix{
        s \ar@{=>}[d]_{H_{[1,2]}} \ar[r]^{X_{[1,2]}} &
        r \ar@{=>}[d]^{H_{[1,2]}} \\
        t \ar@{->}[r]^{\omega^{4}_{[2]}} &
        q
      }
      \xymatrix{
        s \ar@{=>}[d]_{H_{[1,2]}\omega_{[1]}} \ar[rr]^{X_{[1,2]}} &&
        r \ar@{=>}[d]^{H_{[1,2]}\omega^3_{[1]}} \\
        t \ar@{->}[rr]^{\omega^7_{[1]}\omega^{3}_{[2]}X_{[1,2]}} &&
        q
      }
      \xymatrix{
        s \ar@{=>}[d]_{H_{[1,2]}\omega^2_{[1]}} \ar[rr]^{X_{[1,2]}} &&
        r \ar@{=>}[d]^{H_{[1,2]}\omega^2_{[1]}} \\
        t \ar@{->}[rr]^{\omega^6_{[1]}\omega^{2}_{[2]}X_{[1,2]}} &&
        q
      }
      \]
      \[
      \xymatrix{
        s \ar@{=>}[d]_{H_{[1,2]}\omega^3_{[1]}} \ar[rr]^{X_{[1,2]}} &&
        r \ar@{=>}[d]^{H_{[1,2]}\omega_{[1]}} \\
        t \ar@{->}[rr]^{\omega^5_{[1]}\omega_{[2]}X_{[1,2]}} &&
        q
      }
\]
Clockwise from the top, we utilise equation $\eqref{ax-h}$ and then equations $\eqref{ax-z_1}$ through $\eqref{ax-z_3}$ from Tables \ref{tab-equations} and \ref{tab-equations2}. In all cases, $\level(\vec{G}) < \level(s),$ because $N$ is a syllable and the $\omega$ and $X$ operations do not change $\delta$-exponents.

  \subcase\label{subcase-3.2} The pivot column of $s$ is of the form \[ v = \frac{1}{\delta^k}\begin{bmatrix}  
        \omega^{\iota}\,(\delta^3)\\ 
        0\,(\delta)\\
        \vdots\\
        \omega^{h}\,(\delta^3)\\
		x\\
        \vdots\\
        z
      \end{bmatrix}.
      \]
  The modified Giles-Selinger algorithm prescribes the syllable $\vec{N} = H_{[1,\alpha]}\omega^{h-\iota}_{[1]},$ where $\alpha$ corresponds to the second non-zero entry of $\boldsymbol{v}.$ The pivot column of $r = X_{[1,2]}s$ is \[ X_{[1,2]}v = \frac{1}{\delta^k}\begin{bmatrix}  
         0\,(\delta)\\
         \omega^{\iota}\,(\delta^3)\\
         0\,(\delta)\\
         \vdots\\
		 \omega^{h}\,(\delta^3)\\
		 x\\
 		\vdots\\
 		z
      \end{bmatrix}.
      \]
      The modified Giles-Selinger algorithm specifies syllable $r\nedge{H_{[2,\alpha]}\omega^{h-\iota}_{[2]}}q,$ and by the equations $X_{[1,2]}H_{[2,\alpha]}\approx H_{[1,\alpha]}X_{[1,2]}$ and $\omega_{[1]}X_{[1,2]} \approx X_{[1,2]}\omega_{[2]}$ (respectively $\eqref{ax-p}$ and $\eqref{ax-c}$ from Tables \ref{tab-equations2} and \ref{tab-equations}), we can take $\vec{G'} = X_{[1,2]}$ and $\vec{N'} = H_{[2,3]}:$\[ \xymatrix{
        s \ar@{=>}[d]_{H_{[1,\alpha]}\omega^{h-\iota}_{[1]}} \ar[r]^{X_{[1,2]}} &
        r \ar@{=>}[d]^{H_{[2,\alpha]}\omega^{h-\iota}_{[2]}} \\
        t \ar@{->}[r]^{X_{[1,2]}} &
        q
      }
      \]
      Since $X_{[1,2]}$ does not increase $\delta$-exponents, and the normal edges are syllables, it is clear that $\level(t) = \level(q) < \level(s).$ 
      
\subcase The pivot column of $s$ is of the form \[
\boldsymbol{v} = \begin{bmatrix}
0\,(\delta)\\
\omega^\iota\,(\delta^3)\\
0\,(\delta)\\
\vdots\\
\omega^h\,(\delta^3)\\
x\\
\vdots\\
z
\end{bmatrix}.
\]
This case is virtually identical to Subcase~\ref{subcase-3.2}.

\subcase The pivot column of $s$ is of the form \[ \boldsymbol{v} = \frac{1}{\delta^k}\begin{bmatrix}  
         0\,(\delta)\\ 0\,(\delta)\\
         \vdots\\         
         \omega^{\iota}\,(\delta^3)\\
         0\,(\delta)\\
         \vdots\\
          0\,(\delta)\\
          \omega^{h}\,(\delta^3)\\
          x\\          
          \vdots\\
          z          
      \end{bmatrix}.
      \]
      The modified Giles-Selinger algorithm gives reduction step $s\nedge{H_{[\alpha,\beta]}\omega^{h-\iota}_{[\alpha]}}t,$ where $\alpha$ and $\beta$ correspond to the first two non-zero entries of $\boldsymbol{v}.$ Observe that the pivot column of $r = X_{[1,2]}s$ is of the same form as $\boldsymbol{v}$. Thus the modified Giles-Selinger algorithm specifies the syllable $r\nedge{H_{[\alpha,\beta]}\omega^{h-\iota}_{[\alpha]}}q.$ 
      
      Then equations \eqref{ax-e} and \eqref{ax-n} from Table \ref{tab-equations} allow us to take $\vec{G'} = X_{[1,2]}:$\[ \xymatrix{
        s \ar@{=>}[d]_{H_{[3,4]}\omega^{h-\iota}_{[3]}} \ar[r]^{X_{[1,2]}} &
        r \ar@{=>}[d]^{H_{[3,4]}\omega^{h-\iota}_{[3]}} \\
        t \ar@{->}[r]^{X_{[1,2]}} &
        q
      }
      \]
      Since $H_{[3,4]}\omega^{h-\iota}_{[3]}$ is a syllable and $X_{[1,2]}$ does not affect least $\delta$-exponents, we must have $\level(\vec{G'}) < \level(s).$ 
            
\subcase The pivot column of $s$ is of the form \[
      \boldsymbol{v} = \begin{bmatrix}  
		\omega^{\iota}\\
		0\\
		\vdots\\
		0
		\end{bmatrix}.
      \]
Let $\jay$ be the index of the pivot column. There are two cases to consider.

\begin{description}

\subsubcase $j=1.$ Again, we look at the second column $s,$ which is $e_2.$ Thus the second (necessarily pivot) column of $r = X_{[1,2]}s$ is $e_1,$ whence the modified Giles-Selinger algorithm prescribes the syllable $r\nedge{X_{[1,2]}}s.$ Therefore $s\edge{X_{[1,2]}}r$ is retrograde.

\subsubcase\label{sscase-3.5.2} $j > 1.$ Then the modified Giles-Selinger algorithm provides the syllable $s\nedge{\omega^{8-\iota}_{[\jay]}X_{[1,\jay]}}t,$ while the pivot column of $r = X_{[1,2]}s$ is \[
      X_{[1,2]}\boldsymbol{v} = \begin{bmatrix}  
		0\\
		\omega^{\iota}\\
		0\\
		\vdots\\
		0
		\end{bmatrix}.      
      \]
      The modified Giles-Selinger algorithm supplies the syllable $r\nedge{\omega^{8-\iota}_{[\jay]}X_{[2,\jay]}}t.$ \[ 
      \xymatrix{
        s \ar@{=>}[d]_{\omega^{8-\iota}_{[\jay]}X_{[1,\jay]}} \ar[r]^{X_{[1,2]}} &
        r \ar@{=>}[d]^{\omega^{8-\iota}_{[\jay]}X_{[2,\jay]}} \\
        t \ar@{->}[r]^{X_{[1,2]}} &
        q.
      }
      \]
      Equation $\eqref{ax-j}$ from Table \ref{tab-equations} assures us that the diagram commutes. Moreover, $\level(\vec{G'})<\level(s).$
\end{description}   
    
\subcase      
      The pivot column of $s$ is of the form \[
      \boldsymbol{v} = \begin{bmatrix}  
		0\\
		\omega^{\iota}\\
		0\\
		\vdots\\
		0
		\end{bmatrix}
      \]
      This virtually identical to Subcase~\ref{sscase-3.5.2}, as the index of the pivot column is necessarily greater than $1.$

\subcase The pivot column of $s$ is of the form \[
      \boldsymbol{v} = \begin{bmatrix}  
		0\\
		0\\
		\vdots\\
		\omega^{\iota}\\
		0\\
		\vdots\\
		0
		\end{bmatrix}.      
      \]
     Let $\alpha\geq 3$ be the index of the non-zero entry. Then the modified Giles-Selinger algorithm prescribes the reduction step $s\nedge{\omega^{8-\iota}_{[\jay]}X_{[\alpha,\jay]}}t,$ where we again make the local convention that $X_{[\jay,\jay]}=\epsilon,$ to dispose of the case $\alpha=\jay.$ The pivot column of $r = X_{[1,2]}s$ is again $\boldsymbol{v},$ and so the modified Giles-Selinger algorithm provides the reduction step $r\nedge{\omega^{8-\iota}_{[\jay]}X_{[\alpha,\jay]}}t.$ This allows us to take $\vec{G'} = X_{[1,2]}:$\[ 
      \xymatrix{
        s \ar@{=>}[d]_{\omega^{8-\iota}_{[\jay]}X_{[\alpha,\jay]}} \ar[r]^{X_{[1,2]}} &
        r \ar@{=>}[d]^{\omega^{8-\iota}_{[\jay]}X_{[\alpha,\jay]}} \\
        t \ar@{->}[r]^{X_{[1,2]}} &
        q.
      }
      \]
      This provides a complete description of Case~\ref{case-3}.
      
    \end{description}
  \case\label{case-4} $G = X_{[2,3]}.$ 
\begin{description}

\subcase Here we consider $\level(s) = (\jay,\kay,m).$ Suppose $\jay\leq 2;$ then for $\ell>2,$ the $\ell$-th column of $s$ is equal to $e_\ell.$ In this case, it is more convenient to display the whole matrix, rather than simply the pivot column: \[
s = \left[\begin{matrix}
x_1&y_1&0&\cdots&0\\
x_2&y_2&0&\cdots&0\\
x_3&y_3&1&\cdots&0\\
\vdots&\vdots&\vdots&1&\vdots\\
x_n&y_n&0&\cdots&1
\end{matrix}
\right].
\]
Then the edge $s\edge{X_{[2,3]}}r$ results in the following: \[
r = \left[\begin{matrix}
x_1&y_1&0&0&\cdots&0\\
x_2&y_2&0&0&\cdots&0\\
x_4&y_4&0&1&\cdots&0\\
x_3&y_3&1&0&\cdots&0\\
\vdots&\vdots&\vdots&\vdots&1&\vdots\\
x_n&y_n&0&\cdots&0&1
\end{matrix}
\right].
\]
Thus $\level(r) = (\jay+1,0,0);$ the modified Giles-Selinger algorithm specifies the reduction step $r\nedge{X_{[2,3]}}s.$ Therefore $s\edge{X_{[2,3]}}r$ is retrograde, and the result follows.

\subcase The pivot column of $s$ is of the form \[
\boldsymbol{v} = \frac{1}{\delta^k}\begin{bmatrix}
\omega^\iota\,(\delta^3)\\
\omega^h\,(\delta^3)\\
0\,(\delta)\\
x\\
\vdots\\
z
\end{bmatrix}.
\]
The modified Giles-Selinger algorithm prescribes the syllable $H_{[1,2]}\omega^{h-\iota}_{[1]},$ while the pivot column of $r = X_{[2,3]}s$ is \[
X_{[2,3]}\boldsymbol{v} = \frac{1}{\delta^k}\begin{bmatrix}
\omega^\iota\,(\delta^3)\\
0\,(\delta)\\
\omega^h\,(\delta^3)\\
x\\
\vdots\\
z
\end{bmatrix},
\]
and so the modified Giles-Selinger algorithm prescribes the syllable $H_{[1,3]}\omega^{h-\iota}_{[1]}.$ The equations $X_{[2,3]}H_{[1,2]}\approx H_{[1,3]}X_{[2,3]}$ and $\omega_{[1]}X_{[2,3]}\approx X_{[2,3]}\omega_{[1]}$ ($\eqref{ax-o'}$ and $\eqref{ax-e}$ from Table \ref{tab-equations}, respectively) allow us to select $\vec{G'}=X_{[2,3]}$:\[ \xymatrix{
        s \ar@{=>}[d]_{H_{[1,2]}\omega^{h-\iota}_{[1]}} \ar[r]^{X_{[2,3]}} &
        r \ar@{=>}[d]^{H_{[1,3]}\omega^{h-\iota}_{[1]}} \\
        t \ar@{->}[r]^{X_{[2,3]}} &
        q.
      }
      \]
Since $X_{[2,3]}$ does not affect the level of the state and $H_{[1,\alpha]}\omega^{h-\iota}_{[1]}$ are syllables, we have that $\level(t\edge{X_{[2,3]}}q)<\level(s).$

Note that this case is essentially the same as Case~\ref{subcase-3.2}, with a simple permutation of the indices of the equations used. Almost all other cases are directly analogous to Case~\ref{case-2} in the same way; i.e., with suitable permutations of the indices of equations. Below, we detail the only case which substantially differs.

\subcase\label{subcase-4.3} The pivot column of $s$ is of the form \[
	\boldsymbol{v}=\frac{1}{\delta^k}\begin{bmatrix}
	1\,(\delta^3)\\
	1\,(\delta^3)\\
	1\,(\delta^3)\\
	1\,(\delta^3)\\
	\end{bmatrix}.
\]
Let \[
	\boldsymbol{v}=\frac{1}{\delta^k}\begin{bmatrix}
	1+a\delta^3\\
	1+b\delta^3\\
	1+c\delta^3\\
	1+d\delta^3\\
	\end{bmatrix},
\]
where $a,b,c,d\in\Z[\omega].$ 

Observe that it is here, and in particular in Subcase~\ref{bad-case}, that we must restrict our attention to the case $n=4.$ 

The pivot column of $r = X_{[2,3]}r$ is \[
X_{[2,3]}\boldsymbol{v}=\frac{1}{\delta^k}\begin{bmatrix}
	1+a\delta^3\\
	1+c\delta^3\\
	1+b\delta^3\\
	1+d\delta^3\\
\end{bmatrix}.
\]
The modified Giles-Selinger algorithm specifies the syllable $H_{[1,2]}$ for \textit{both} $s$ and $r,$ though the results of these reductions will in general be distinct. Thus our analysis must be refined. Our strategy will be to explore a second application of the modified Giles-Selinger algorithm to each of our states: \[ \xymatrix{
        s \ar@{=>}[d]_{H_{[1,2]}} \ar[r]^{X_{[2,3]}} &
        r \ar@{=>}[d]^{H_{[1,2]}} \\
        t \ar@{=>}[d]_{H_{[3,4]}} &
        q \ar@{=>}[d]^{H_{[3,4]}} \\
        t' \ar@{->}[r]^{\vec{G''}}&
        q'.
      }
      \]
It is clear that completing the above diagram will suffice, since if we can find an appropriate value for $\vec{G''}$ which follows from our equations, we can conclude that $\vec{G'} = t\edge{H_{[3,4]}\vec{G''}H_{[3,4]}}q$ will satisfy the statement of the lemma. Note that the pivot column of $t'$ is \[\boldsymbol{v_{t'}}=\frac{1}{\delta^k\sqrt{2}}\begin{bmatrix}
	2 + (a+b)\delta^3\\
	(a-b)\delta^3\\
	2 + (c+d)\delta^3\\
	(c-d)\delta^3\\
\end{bmatrix},\]
while the pivot column of $q'$ is \[\boldsymbol{v_{q'}}=\frac{1}{\delta^k\sqrt{2}}\begin{bmatrix}
	2 + (a+c)\delta^3\\
	(a-c)\delta^3\\
	2 + (b+d)\delta^3\\
	(b-d)\delta^3\\
\end{bmatrix}.\]
In the following calculation, recall that $\sqrt{2}=\delta^2\omega^{-1}\lambda^{-1}$ and so $2 = \delta^4\omega^{-2}\lambda^{-2}.$ Let $D = \frac{2}{\delta^3} = \delta\omega^{-2}\lambda^{-2},$ and note that $D\equiv\delta\pmod{\delta^2.}$ Thus we have \[\boldsymbol{v_{t'}}=\frac{\lambda\omega}{\delta^{k+2}}\begin{bmatrix}
	2 + (a+b)\delta^3\\
	(a-b)\delta^3\\
	2 + (c+d)\delta^3\\
	(c-d)\delta^3\\
\end{bmatrix} = \frac{\lambda\omega}{\delta^{k-1}}\begin{bmatrix}
	D + (a+b)\\
	(a-b)\\
	D + (c+d)\\
	(c-d)\\
\end{bmatrix}.\]
Similarly \[\boldsymbol{v_{q'}}=\frac{\lambda\omega}{\delta^{k-1}}\begin{bmatrix}
	D + (a+c)\\
	(a-c)\\
	D + (b+d)\\
	(b-d)\\
\end{bmatrix}.\]

We are now ready to delve into further subcases.

\begin{description}
\subsubcase $a+b+c+d \equiv0\pmod{\delta^2}.$ 

Note that $H_{[2,4]}H_{[1,3]}t'$ has pivot column \[
H_{[2,4]}H_{[1,3]}\boldsymbol{v_{t'}} = \frac{\lambda\omega}{\delta^{k-1}\sqrt{2}}\begin{bmatrix}
  2D + (a+b+c+d)\\
  (a-b+c-d)\\
  (a+b-c-d)\\
  (a-b-c+d)\\
\end{bmatrix},
\]
while $H_{[2,4]}H_{[1,3]}q'$ has pivot column \[
H_{[2,4]}H_{[1,3]}\boldsymbol{v_{q'}} = \frac{\lambda\omega}{\delta^{k-1}\sqrt{2}}\begin{bmatrix}
  2D + (a+b+c+d)\\
  (a+b-c-d)\\
  (a-b+c-d)\\
  (a-b-c+d)\\
\end{bmatrix}.
\]
It is clear that the above two vectors differ only by $X_{[2,3]};$ i.e., \[ \xymatrix{
        s \ar@{=>}[d]_{H_{[1,2]}} \ar[r]^{X_{[2,3]}} &
        r \ar@{=>}[d]^{H_{[1,2]}} \\
        t \ar@{=>}[d]_{H_{[3,4]}} &
        q \ar@{=>}[d]^{H_{[3,4]}} \\
        t' \ar@{->}[d]_{H_{[2,4]}H_{[1,3]}}&
        q' \ar@{->}[d]^{H_{[2,4]}H_{[1,3]}} \\
        t'' \ar@{->}[r]^{X_{[2,3]}}
       &
      q''.
      }
      \]
The above diagram commutes thanks to equation $\eqref{ax-s}$ from Table \ref{tab-equations2}. Note also that $2D$ and $(a+b+c+d)$ are divisible by $\delta^2,$ so the least $\delta$-exponents of $t''$ and $q''$ are less than that of $s.$ Clearly the levels of $t,t',q,$ and $q'$ are also less than that of $s,$ whence we are assured that $\vec{G'}= t\edge{H_{[3,4]}H_{[2,4]}H_{[1,3]}X_{[2,3]}H_{[2,4]}H_{[1,3]}H_{[3,4]}}q$ is an appropriate choice. 

\subsubcase $a+b+c+d = \delta\pmod{\delta^2}.$ This case is similar to the above; to wit,
\[
H_{[1,4]}H_{[2,3]}\boldsymbol{v_{t'}} = \frac{\lambda\omega}{\delta^{k-1}\sqrt{2}}\begin{bmatrix}
D + (a+b+c-d)\\
D + (a-b+c+d)\\
-D + (a-b-c-d)\\
D + (a+b-c+d))\\
\end{bmatrix},\]
and \[
H_{[1,4]}H_{[2,3]}\boldsymbol{v_{q'}} = \frac{\lambda\omega}{\delta^{k-1}\sqrt{2}}\begin{bmatrix}
 D + (a+b+c-d)\\
 D + (a+b-c+d)\\
 -D + (a-b-c-d)\\
 D + (a-b+c+d))\\
 \end{bmatrix}.
 \]
Observe that the two states differ by $X_{[2,4]}:$\[ \xymatrix{
        s \ar@{=>}[d]_{H_{[1,2]}} \ar[r]^{X_{[2,3]}} &
        r \ar@{=>}[d]^{H_{[1,2]}} \\
        t \ar@{=>}[d]_{H_{[3,4]}} &
        q \ar@{=>}[d]^{H_{[3,4]}} \\
        t' \ar@{->}[d]_{H_{[1,4]}H_{[2,3]}}&
        q' \ar@{->}[d]^{H_{[1,4]}H_{[2,3]}} \\
        t'' \ar@{->}[r]^{X_{[2,4]}}
       &
      q''.
      }
      \]
      The correctness of the above diagram follows from equation $\eqref{ax-t}$ of Table \ref{tab-equations2}. Moreover, since $D + (a+b+c+d)$ is divisible by $\delta^2$, the
levels of $t''$ and $q''$ are as they
should be.

\subsubcase\label{bad-case} $ a+b+c+d = 1\pmod{\delta}.$ Recall from above that \[
\boldsymbol{v_{t'}}=\frac{\lambda\omega}{\delta^{k-1}}\begin{bmatrix}
	D + (a+b)\\
	(a-b)\\
	D + (c+d)\\
	(c-d)\\
\end{bmatrix}.
\]
Let $x = D+a+b,$ $y = a - b,$ $z = D+c+d,$ and $w = c-d.$ Note that $x,y,z,w$ are all distinct modulo $\delta^2.$ 

Since $\lambda\omega$ is a unit of the ring under consideration, it follows that $\lambda\omega x,\lambda\omega y, \lambda\omega z,$ and $\lambda\omega w$ are also distinct modulo $\delta^2.$ Therefore these values encompass $0,1,1+\delta,$ and $\delta$ (all $(\bmod~\delta^2)$) in some order. But the vector \[
\boldsymbol{v_{t'}} = \frac{1}{\delta^{k-1}}\begin{bmatrix}
\lambda\omega x\\
\lambda\omega y\\
\lambda\omega z\\
\lambda\omega w
\end{bmatrix}
\]
is unitary, directly contradicting Lemma \ref{NotExist4DimLemma}. This shows that Subcase~\ref{bad-case} cannot arise.

\end{description}

\subcase The pivot column of $s$ is of the form \[
	\boldsymbol{v}=\frac{1}{\delta^k}\begin{bmatrix}
	\omega^\iota\,(\delta^3)\\
	\omega^h\,(\delta^3)\\
	\omega^\ell\,(\delta^3)\\
	\omega^\tau\,(\delta^3)\\
	\end{bmatrix}.
\]
We will reduce this to Subcase~\ref{subcase-4.3} as follows.

The modified Giles-Selinger algorithm  specifies the syllable $s\nedge{H_{[1,2]}\omega^m_{[1]}}s',$ where $m\equiv h-\iota\pmod{4}.$ Now consider $\vec{G_1} = s'\edge{\omega^{-h}_{[1]}\omega^{-h}_{[2]}\omega^{-\ell}_{[3]}\omega^{-\tau}_{[4]}}t.$ To compute the pivot column $\boldsymbol{v_t}$ of $t,$ we use \[
\begin{array}{ccc}
\boldsymbol{v_t}&=&\omega^{-h}_{[1]}\omega^{-h}_{[2]}\omega^{-\ell}_{[3]}\omega^{-\tau}_{[4]}\boldsymbol{v_{s'}}\\
&=&\omega^{-h}_{[1]}\omega^{-h}_{[2]}\omega^{-\ell}_{[3]}\omega^{-\tau}_{[4]}H_{[1,2]}\omega^m_{[1]}\boldsymbol{v}\\
&=&H_{[1,2]}\omega^{m-h}_{[1]}\omega^{-h}_{[2]}\omega^{-\ell}_{[3]}\omega^{-\tau}_{[4]}\boldsymbol{v}.
\end{array}
\]
Here we have used equations $\eqref{ax-w},\eqref{ax-b}$ and $\eqref{ax-f'}$ from Table \ref{tab-equations} to derive the final equality.

But $m-h\equiv-\iota\pmod{4},$ so \[\omega^{m-h}_{[1]}\omega^{-h}_{[2]}\omega^{-\ell}_{[3]}\omega^{-\tau}_{[4]}\boldsymbol{v} = \begin{bmatrix}
1\,(\delta^3)\\
1\,(\delta^3)\\
1\,(\delta^3)\\
1\,(\delta^3)\\
\end{bmatrix}.
\]
Therefore \[\boldsymbol{v_t} = H_{[1,2]}\begin{bmatrix}
1+a\delta^3\\
1+b\delta^3\\
1+c\delta^3\\
1+d\delta^3\\
\end{bmatrix},
\]
which is precisely the situation obtained in Subcase~\ref{subcase-4.3}. Observe also that the pivot column of $r$ is \[
\boldsymbol{v_r} = X_{[2,3]}\boldsymbol{v}=\frac{1}{\delta^4}\begin{bmatrix}
	\omega^\iota\,(\delta^3)\\
	\omega^\ell\,(\delta^3)\\
	\omega^h\,(\delta^3)\\
	\omega^\tau\,(\delta^3)\\
\end{bmatrix}.
\]
The modified Giles-Selinger algorithm specifies the syllable $r\nedge{H_{[1,2]}\omega^\phi_{[1]}}r',$ where $\phi\equiv \ell-\iota\pmod{4}.$

Consider $r'\edge{\omega^{-\ell}_{[1]}\omega^{-\ell}_{[2]}\omega^{-h}_{[1]}\omega^{-\tau}_{[1]}}q.$ Using exactly the same calculation as for $s,$ we find that \[
\boldsymbol{v_q} = H_{[1,2]}\begin{bmatrix}
1+a\delta^3\\
1+c\delta^3\\
1+b\delta^3\\
1+d\delta^3\\
\end{bmatrix}.
\]
Therefore, we end up in the same position as in Subcase~\ref{subcase-4.3}; to wit:\[ \xymatrix{
        s \ar@{=>}[d]_{H_{[1,2]}\omega^m_{[1]}} \ar[r]^{X_{[2,3]}} &
        r \ar@{=>}[d]^{H_{[1,2]}\omega^\phi_{[1]}} \\
        s' \ar@{->}[d]_{\omega^{-h}_{[1]}\omega^{-h}_{[2]}\omega^{-\ell}_{[3]}\omega^{-\tau}_{[4]}} &
        r' \ar@{->}[d]^{\omega^{-\ell}_{[1]}\omega^{-\ell}_{[2]}\omega^{-h}_{[3]}\omega^{-\tau}_{[4]}} \\
        t&
        q,
      }
      \]
      where $t,q$ are exactly as in Subcase~\ref{subcase-4.3}. Thus the current case is completed precisely as in that case.
\end{description}      
  \case\label{case-5} $G = X_{[\alpha,\alpha+1]}$ for $\alpha > 2.$ This case is entirely analogous to Case~\ref{case-4}, as all arising subcases differ from Case~\ref{case-4} at most by a monotone re-indexing of rows.
  
  \end{description}
  
  With that, we have exhausted all of the basic edges, and thus have completed the proof of the lemma.

\chapter{Conclusion}

The main result of this thesis is Theorem \ref{MainThm}, which provides a presentation, in terms of generators and relations, of the group $U_4(\D[\omega])$ of unitary operators that are described by matrices whose entries come from the ring $\D[\omega].$ We achieved this result by modifying the algorithm from \cite{GILES-SEL}, which gave us prospective syllables on words in the $(*)$ generators for $U_4(\D[\omega]).$ We produced the equations from Tables \ref{tab-equations} and \ref{tab-equations2} by experimentation, particularly during the proof of Lemma \ref{TechLemma}. We then showed that any word in the $(*)$ generators is equationally equivalent to the unique string of syllables output by Algorithm \ref{GSAlg}. This yields the desired presentation of the group $U_4(\D[\omega])$.

Our result has a potential application in quantum mechanics, namely the simplification of Clifford+$T$ circuits for two qubits, as presented in the introduction. A further step would therefore be to translate the equations from Tables \ref{tab-equations} and \ref{tab-equations2} into equations about the Clifford+$T$ generators. 

While the complexity of Algorithm \ref{GSAlg} is not directly relevant to the results of this work (we only require the algorithm to terminate), it is nevertheless of some independent interest. In \cite{GILES-SEL}, Giles and Selinger show that their algorithm produces $O(3^nk)$ one- and two-level matrices, where $n$ refers to the dimension of the input matrix and $k$ refers to the least denominator exponent of that matrix. Our modification works instead with $\delta$ exponents, but the complexity argument is otherwise identical.

Some questions have been left open. Most notably, we have only treated the case $n=4.$ In fact, Subcase~\ref{bad-case} of the proof of Lemma \ref{TechLemma} was the only point within the proof architecture of Theorem \ref{MainThm} which did not readily admit generalisation to the case $n>4.$ Such generalisation could be applied to the Clifford+$T$ group on more than 2 qubits. Another question we did not address is whether the equations from Tables \ref{tab-equations} are independent; it may well be the case that some equations are consequences of the others. We leave these questions to future work.

\appendix
\bibliographystyle{abbrv}
  \bibliography{SethMastersThesis}
\end{document}